%% file: main.tex
\documentclass{article}

\usepackage{iclr2027_conference,times}
\usepackage{amsmath,amssymb,amsthm}
\usepackage{hyperref}
\usepackage{url}
\usepackage{booktabs}
\usepackage{multirow}
\usepackage{algorithm}
\usepackage{algpseudocode}
\usepackage{xspace}
\usepackage{tablefootnote}
\usepackage{threeparttable}
\usepackage{ashwin}
\usepackage{inconsolata}
\usepackage[most]{tcolorbox}
\usepackage{xcolor}

\newtcolorbox[
    auto counter
]{algorithmbox}[2][]{
    enhanced,
    float,
    floatplacement=tbp,
    colback=white,
    colframe=gray!48!black,
    colbacktitle=gray!20,
    coltitle=black,
    fonttitle=\bfseries,
    title={Algorithm~\thetcbcounter: #2},
    boxrule=0.5pt,
    titlerule=0.3pt,
    arc=1.5mm,
    left=10pt,
    right=10pt,
    top=8pt,
    bottom=8pt,
    toptitle=5pt,
    bottomtitle=5pt,
    #1
}

\newcommand{\build}{\mathsf{Build}}
\newcommand{\query}{\mathsf{Query}}
\newcommand{\norm}[1]{\lVert #1 \rVert}

\newcommand{\costo}{C_{\textup{oracle}}}
\newcommand{\costr}{C_{\textup{retrieval}}}

\newcommand{\lift}{\textsc{SoftmaxLift}\xspace}
\newcommand{\adaptivebuckets}{\textsc{AdaptiveBuckets}\xspace}
\newcommand{\topk}{\textsc{TopK}\xspace}
\newcommand{\topksampling}{\textsc{TopK+Sampling}\xspace}

\DeclareMathOperator{\attn}{Attn}

\newtheorem{theorem}{Theorem}
\newtheorem{lemma}[theorem]{Lemma}
\newtheorem{claim}[theorem]{Claim}
\theoremstyle{definition}
\newtheorem{definition}{Definition}
\theoremstyle{remark}

\title{Attention via Black-Box Vector Search}
\author{%
\normalfont
\begin{tabular}[t]{@{}l@{\hspace{8em}}l@{}}
\textbf{Stepan Zharkov}\footnotemark[1] &
\textbf{Krish Singal}\footnotemark[1] \\
Columbia University &
University of Pennsylvania \\
\texttt{styopa@cs.columbia.edu} &
\texttt{ksingal@engineering.upenn.edu} \\[1.5em]
\textbf{Ashwin Padaki}\footnotemark[1] &
\textbf{Alexandr Andoni} \\
University of Pennsylvania &
Columbia University \\
\texttt{apadaki@engineering.upenn.edu} &
\texttt{andoni@cs.columbia.edu}
\end{tabular}%
}

\iclrfinalcopy

\begin{document}

\maketitle

\begingroup
\renewcommand{\thefootnote}{\fnsymbol{footnote}} 
\footnotetext[1]{Equal contribution, random order.} 
\endgroup                                       

\begin{abstract}
    Sparse attention mechanisms estimate attention over $n$ tokens using a small subset of keys. Many existing approaches use maximum inner product search (MIPS) to retrieve the heaviest keys, which motivates the following question: given \emph{black-box} access to a MIPS oracle, how many keys must be retrieved to output an $\eps$-accurate attention estimate?

    We answer this question by unifying prior approaches through the framework of \emph{priority sampling}. With a single MIPS index, we show that $\Theta(\sqrt{n}/\eps)$ retrieved keys are both sufficient and necessary. With $\Theta(\log n)$ indices, we give an algorithm that retrieves only $O(\log n+1/\eps^2)$ keys and prove that this is near-optimal. More generally, we design algorithms that establish a smooth tradeoff between the number of MIPS indices and number of retrieved keys. We then show that if we allow augmentation of keys and queries, we can bypass the above lower bounds: there exists a simple priority-sampling estimator using a single MIPS index and $O(1/\eps^2)$ retrieved keys. When integrated into LLM inference, our algorithms outperform top-$k$ and sampling approaches used in prior work and yield attention approximation that scales favorably to long contexts.
\end{abstract}

\input{sections/intro.tex}

\input{sections/algorithm.tex}

\input{sections/experiments.tex}
\input{sections/limitations.tex}

\bibliography{main}
\bibliographystyle{iclr2027_conference}

\appendix

\input{sections/algorithms-appendix}

\input{sections/lower-bound.tex}

\input{sections/experiments-appendix}

\end{document}

%% file: sections/intro.tex
\section{Introduction}\label{sec:intro}

Transformer-based large language models (LLMs) increasingly rely on long context windows to process documents, codebases, and extended reasoning traces. A major bottleneck to scaling these context windows is the \emph{attention mechanism}~\citep{vaswani2017attention}. Given an existing context of $n$ tokens, computing attention for each new token requires comparing it against all preceding tokens. Formally, given key vectors $x_1,\ldots,x_n\in\R^d$, value vectors $v_1,\ldots,v_n\in\R^{D}$, and a query vector $q\in\R^d$, the attention mechanism computes \[\attn(q) := \sum_{i=1}^n \frac{\kappa(q,x_i)}{\sum_{j=1}^n \kappa(q,x_j)} \cdot v_i,\] where $\kappa(q,x) = \exp(\ip{q}{x})$ is the softmax kernel.\footnote{More generally, $\kappa(q,x) = \exp(\ip{q}{x}/\tau)$ for a temperature parameter $\tau>0$.} 

Recent work circumvents the cost of computing the attention mechanism using the approach of \emph{sparse attention}~\citep{kitaev2020reformer,zaheer2020bigbird}. Rather than attending to every token, sparse attention mechanisms identify a small subset of tokens and use them to approximate the true attention vector. The attention weight of token $i$ grows exponentially with $\ip{q}{x_i}$, suggesting a natural strategy: retrieve the keys that have the largest inner product with the query. This is precisely a \emph{maximum inner product search} (MIPS) problem, a well-studied primitive in vector search~\citep{shrivastava2014alsh,neyshabur2015symmetric}. Indeed, a growing body of work builds sparse attention mechanisms with MIPS as the core subroutine, including RetrievalAttention~\citep{retrievalattention}, HashAttention~\citep{hashattention}, and vAttention~\citep{vattention}.

A common feature of MIPS-based sparse attention mechanisms is that they generally treat the underlying MIPS algorithm as a \emph{black box}: the estimator specifies which keys to retrieve and how to utilize them, but is agnostic to how those keys are found.\footnote{Notable exceptions include MagicPig~\citep{magicpig} and SAAP~\citep{saap}, whose attention estimates are coupled to the search algorithms they use.} Using MIPS as a black box is attractive because vector search is already a highly optimized problem. Decades of algorithmic and engineering progress have produced extremely performant approximate nearest neighbor and MIPS algorithms, as documented in public benchmarks such as ANN-Benchmarks~\citep{aumuller2020annbenchmarks} and VIBE~\citep{jaasaari2026vibe}. Any sparse attention method that is agnostic to the underlying search implementation can immediately benefit from state-of-the-art algorithms and future improvements.

In this paper, we take a theoretical perspective on MIPS-based sparse attention. We assume access to a perfect MIPS oracle and study the \emph{retrieval strategy}: the algorithm deciding which keys to retrieve and how to use the retrieved keys to estimate attention. We begin with the following question: \begin{center}
    \emph{Given a perfect MIPS oracle, what is the optimal retrieval strategy\\for estimating attention within provably low error?}
\end{center}

We make substantial progress on this question by systematically characterizing the tradeoff between the number of MIPS indices used by the algorithm and the number of keys it retrieves for an $\eps$-accurate estimate. Some points on this tradeoff were already known for the related problem of \emph{kernel density estimation} (KDE), and these results extend to attention: \cite{mussmann2017fast} gives a one-index algorithm retrieving $O(\sqrt{n}/\eps)$ keys, while \cite{mussmann2026sum} gives an algorithm using $O(\log n)$ indices and retrieving only $O(\log n / \eps^2)$ keys.

We unify these results under a common algorithmic framework of \emph{priority sampling}~\citep{duffield2007priority,daliri2024simple}, and we generalize them significantly. For a budget $T$ on the number of indices kept, we give adaptive and non-adaptive algorithms achieving precise oracle--retrieval tradeoffs. In the two oracle budget settings studied in prior work, we obtain a tight answer: for $T=1$, we prove that $O(\sqrt{n}/\eps)$ retrieved keys is optimal up to constant factors; for $T=\Theta(\log n)$ indices, we give an adaptive algorithm retrieving $O(\log n + 1/\eps^2)$ keys and prove that it is optimal up to a $\log\log n$ factor. Motivated by the strong lower bounds, we show that a slight relaxation of our model allows priority sampling to be implemented much more efficiently: lifting keys and queries into one extra dimension, our algorithm \lift retrieves only $O(1/\eps^2)$ keys, independent of the context length $n$. 

We also explore the practicality of \lift. We implement it in Llama-3.2-1B-Instruct using fast approximate MIPS oracles and evaluate it both on per-head attention error and end-to-end benchmarks. \lift substantially improves the accuracy--latency tradeoff over the top-$k$ and sampling approaches common in prior work. Finally, we provide preliminary evidence that \lift can enable efficient sparse attention at context lengths of up to 1M tokens.

\subsection{Attention in the MIPS Model}
\label{subsec:black-box}

We consider a model for computing attention given access to a black-box MIPS oracle. 

\begin{definition}[MIPS oracle]\label{def:oracle}
    For $x_1,\ldots,x_n\in\R^d$, a \emph{MIPS oracle} supports the following operations: \begin{itemize}
        \item $\build(S)$ creates a MIPS index on a subset $S\subseteq [n]$.
        \item $\query(q,k;S)$ returns the $\min\{k,|S|\}$ pairs $(i,\ip{q}{x_i})$ with largest inner products in $S$.\footnote{We assume throughout the paper that ties are broken lexicographically.}
        \end{itemize}
\end{definition}

\begin{definition}[MIPS model]
A \emph{MIPS attention algorithm} $\calA$ operates in the following access model: \begin{itemize}
    \item \textbf{Preprocessing.} $\calA$ receives the key and value vectors, selects subsets $S_1,\ldots S_T\subseteq[n]$, and runs $\build(S_t)$ for $t\in[T]$ to instantiate MIPS indices.
    \item \textbf{Query stage.} $\calA$ receives a query $q\in\R^d$, but its access is limited to the following operations: it can run $\query(q,k;S_t)$ for $t\in[T]$, or evaluate $\ip{q}{x_i}$ for an arbitrary $i\in[n]$. By the end of its computation, algorithm $\calA$ outputs an estimate $\smash{\widehat{\attn}(q)}$ of the attention.
\end{itemize}

$\calA$ is \emph{non-adaptive} if, during the query stage, all MIPS queries and inner product evaluations are made simultaneously; otherwise, $\calA$ is \emph{adaptive}. The complexity of $\calA$ is given by two metrics: \begin{itemize}
    \item \textbf{Oracle cost:} $\costo(\calA) = T$, the number of MIPS indices constructed.
    \item \textbf{Retrieval cost:} $\costr(\calA) = \left(\sum k\right) + r$, where the summation is over all MIPS queries $\query(q,k;\cdot)$ made by $\calA$, and $r$ is the number of direct inner product evaluations.
\end{itemize}
\end{definition}

\begin{definition}[Accuracy guarantee]\label{def:accuracy}
    An attention algorithm $\calA$ is \emph{$\eps$-accurate} if, for every dataset of keys and values and every fixed query $q$, \begin{equation}\label{eq:main-attention-error}
    \Pr\left[
        \norm{\widehat{\attn}(q)-\attn(q)}_2
        \leq \eps\cdot\max_{i\in[n]}\,\norm{v_i}_2
    \right]\geq\frac{2}{3},
\end{equation} where the probability is over the randomness of $\calA$. We note that measuring error with respect to the maximum value vector norm is common in prior work~\citep{kochetkova2026streaming,schroder2026wildcat,liberty2026coresets} and is natural when, e.g., the value vectors have unit norm.
\end{definition}

\subsection{Theoretical Results}\label{subsec:theoretical-results}

We provide adaptive and nonadaptive algorithms achieving the following tradeoffs between retrieval cost and oracle cost.

\begin{table*}[t]
\centering
\small
\setlength{\tabcolsep}{7pt}
\renewcommand{\arraystretch}{1.35}

\begin{tabular}{@{}ccccc@{}}
\toprule
\multirow{2}{*}{\textbf{Oracle budget}}
& \textbf{Adaptive}
& \textbf{Non-adaptive}\footnotemark
& \textbf{Lower bound}
& \textbf{Softmax Lift} \\
& {\footnotesize (Theorem~\ref{thm:bucket-adaptive})}
& {\footnotesize (Theorem~\ref{thm:bucket-nonadaptive})}
& {\footnotesize (Theorems~\ref{thm:lower-bound-const},~\ref{thm:lower-bound-any})}
& {\footnotesize (Theorem~\ref{thm:softmax-lift})} \\
\midrule

$T=1$
& $O(\sqrt{n}/\eps)$
& $O(\sqrt{n}/\eps)$
& $\Omega(\sqrt{n}/\eps)$
& $O(\eps^{-2})$ \\

$T=2$
& $O(n^{1/3}\eps^{-4/3})$
& $O(n^{1/3}\eps^{-4/3})$
& $\Omega\!\left(
    \frac{\log n}{\log\log n}+\eps^{-2}
  \right)$
& --- \\

$T=\lceil\log_2 n\rceil$
& $O(\log n+\eps^{-2})$
& $O(\log n/\eps^2)$
& $\Omega\!\left(
    \frac{\log n}{\log\log n}+\eps^{-2}
  \right)$
& --- \\

\bottomrule
\end{tabular}

\caption{Number of retrieved keys for $\eps$-accurate attention for oracle budgets $T$.
All lower bounds are implicitly capped at $n$.
The final column performs an augmentation to the keys and query and is not subject to the lower bounds.}
\label{tab:regimes}
\end{table*}

\footnotetext{We attribute the non-adaptive results for $T=1$ and $T=\Theta(\log n)$ to~\cite{mussmann2017fast,mussmann2026sum}.}

\begin{theorem}[Adaptive algorithm]\label{thm:bucket-adaptive}
    For any oracle budget $T\geq1$ and parameter $\eps\in(0,1)$, there is an $\eps$-accurate adaptive MIPS attention algorithm $\calA$ with $\costo(\calA)\le T$ and 
    \[\costr(\calA) \le
        O\!\left(T+n^{1/(T+1)}\eps^{-2+2/(T+1)}\right).
    \]
\end{theorem}

\begin{theorem}[Non-adaptive algorithm]\label{thm:bucket-nonadaptive}
    For any oracle budget $T\geq1$ and parameter $\eps\in(0,1)$, there is an $\eps$-accurate non-adaptive MIPS attention algorithm $\calB$ with $\costo(\calB)\le T$ and 
    \[\costr(\calB) \le
        O\!\left(Tn^{1/(T+1)}\eps^{-2+2/(T+1)}\right).
    \]
\end{theorem}

We show that the above tradeoffs are tight for $T=1$, even up to constant factors.

\begin{theorem}[Lower bound for $T=1$]\label{thm:lower-bound-const}
For sufficiently small $\eps>0$, any $\eps$-accurate MIPS attention algorithm with $\costo(\calA) \le 1$ has retrieval cost
\[
    \Omega\left(\min\{n,\, \sqrt{n}/\eps\}\right).
\]
\end{theorem}

We also give a lower bound for arbitrary $T$ showing that the $O(\log n + 1/\eps^2)$ bound achieved by our adaptive algorithm is nearly tight for $T=\Theta(\log n)$.

\begin{theorem}[Lower bound for any $T$]\label{thm:lower-bound-any}
For sufficiently small $\eps>0$, any $\eps$-accurate MIPS attention algorithm must have \[\costr(\calA)\ge \Omega\left(\min\left\{n,\,\frac{\log n}{\log\log n} + \frac{1}{\eps^2}\right\}\right).\]
\end{theorem}

Finally, we show that if we slightly extend the model to allow augmenting the keys (and the query) before constructing MIPS indices, we can achieve a much stronger upper bound.

\begin{theorem}[Softmax lift]\label{thm:softmax-lift}
    There is an $\eps$-accurate attention algorithm $\calZ$ that, after lifting the keys and queries into $\R^{d+1}$, has $\costo(\calZ) = 1$ and $\costr(\calZ) = O(1/\eps^2)$.
\end{theorem}

\subsection{Empirical Results}

We test the empirical performance of our algorithms by integrating them into  Llama-3.2-1B-Instruct long context inference. We design a custom harness that allows for combining arbitrary retrieval strategies with arbitrary MIPS oracles. For each combination of retrieval strategy and oracle, the harness measures task accuracy and latency on a variety of 128k context tasks in the RULER suite~\citep{hsieh2024ruler}. We benchmark the following retrieval strategies:
\begin{itemize}
    \item \topk . Request the top $k$ keys from the oracle and use the partial attention sum on these keys as the estimator (e.g. ~\cite{gupta2021topk,retrievalattention}).
    \item \topksampling. Request $k'$ keys from the oracle, and uniformly sample $k-k'$ more keys. Compute the partial attention sum with the sampled keys weighted to create an unbiased estimator (e.g.~\cite{vattention}).
    \item \lift, the priority sampling based strategy introduced in Algorithm~\ref{alg:softmax-lift-alg}. 
    \item \adaptivebuckets, the priority sampling based strategy introduced in Algorithm \ref{alg:adaptive-alg}.
\end{itemize}

\begin{figure}[t]
    \centering
    \resizebox{\linewidth}{!}{%
        \input{tikz-diagrams/flat-k-vs-quality}%
    }
    \caption{Quality as a function of the retrieval budget $k$ with an
    exact FAISS-Flat MIPS oracle on 60 matched 128K-context RULER examples from 10 tasks. Top-$k$+sampling divides its budget
    evenly between retrieved and uniformly sampled keys. }
    \label{fig:flat-k-vs-quality}
\end{figure}
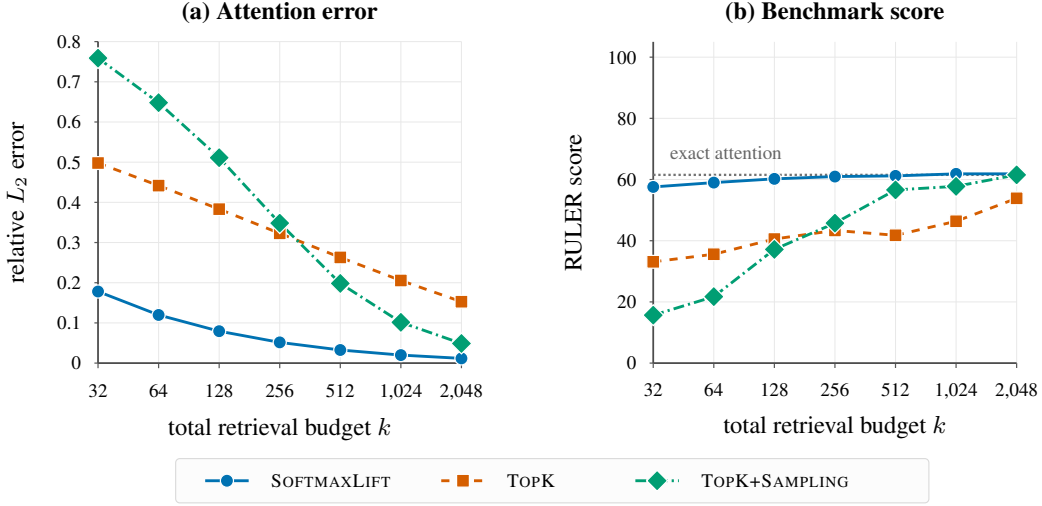

\lift outperforms the \topk and \topksampling retrieval strategies commonly used in prior work (Figure \ref{fig:flat-k-vs-quality}).  We also evaluate our approach on a small set of 1M context examples with the Llama-3.1-Nemotron-8B-UltraLong-1M-Instruct model. These experiments provide preliminary evidence that \lift scales well to even larger contexts. See Section \ref{sec:experiments} and Appendix \ref{sec:exp-apx} for the complete set of experimental results. 

\subsection{Related Work}\label{subsec:related}

\paragraph{Sparse attention.} Early sparse attention methods fixed the sparsity pattern in advance. Examples include strided or block-local attention~\citep{child2019sparsetransformer}, local windows with a small set of global tokens~\citep{beltagy2020longformer}, and mixtures of local, global, and random blocks~\citep{zaheer2020bigbird}. Reformer~\citep{kitaev2020reformer} takes a more adaptive approach, using locality sensitive hashing (LSH) to bucket queries and keys. A separate line of work evicts tokens from the
KV cache based on accumulated attention mass on statistics collected over a calibration window~\citep{zhang2023h2o,liu2023scissorhands,ge2024fastgen,li2024snapkv,xiao2024streamingllm}. The drawback is that eviction is irreversible; a token cannot be recovered if a later query
needs it. This limitation motivated query-aware methods that retain the full cache and select a subset per
query. Different methods use page-level summaries~\citep{tang2024quest,xiao2024infllm}, channel
sparsification~\citep{yang2024doublesparsity}, or low-rank key
projections~\citep{singhania2024loki}. More recent MIPS-based mechanisms (discussed in Section~\ref{sec:intro}) build graph, clustering, quantization, or hashing indices over the
keys~\citep{retrievalattention,squeezedattention,zhang2024pqcache,hashattention,saap}. A
recurring obstacle for these approaches is that attention queries and keys follow different
distributions. The same out-of-distribution phenomenon is
studied in the vector search literature~\citep{jaiswal2022ooddiskann,chen2024roargraph}, and it is one reason we treat the oracle as a black box. Finally, sparsity can also be built directly into the model at training time~\citep{yuan2025nsa}, a direction orthogonal to ours.

\paragraph{Beyond top-$k$.} The standard top-$k$ approximation approach retrieves the $k$ largest attention weights and renormalizes over them. Because it drops all remaining softmax mass, the reuslting estimator is biased.~\citet{magicpig} showed empirically that even \emph{exact} top-$k$ is beaten at equal budget by sampling keys in proportion to their attention weight. The work of~\citep{zhu2025tactic} builds on this by keeping the estimator but adapting the budget, retrieving enough keys to cover a $p$-fraction of the mass. Further work splits the distribution into a ``heavy'' head which is computed exactly and a tail which is sampled~\citep{vattention}. A similar decomposition appears in~\cite{han2024hyperattention,karppa2022deann}.

\paragraph{Priority sampling and sampling by search.}
Priority sampling~\citep{duffield2007priority}, known in statistics as sequential Poisson sampling~\citep{ohlsson1998sequential}, draws a fixed-size sample supporting unbiased subset-sum estimation through Horvitz--Thompson
weights~\citep{horvitz1952generalization,alon2005subsetsums}. Its variance is optimal in a strong
sense~\citep{szegedy2006priority,daliri2024simple}, and it has recently been applied to inner
product sketching~\citep{daliri2024inner}. The observation behind our \lift algorithm is that a priority
sample is determined by a top-$k$ operation on perturbed weights, which is exactly what a MIPS oracle computes. Related ideas appear in the Gumbel-max literature, where \citet{mussmann2016mips,mussmann2017fast} sample from a log-linear model by perturbing scores with Gumbel noise and solving a MIPS instance, and the Gumbel-top-$k$ trick extends this to sampling without replacement~\citep{maddison2014astar,kool2019gumbel}. \citet{spring2017lsh}
and~\citet{luo2018ace} instead treat the LSH table itself as the sampler; this is also the perspective adopted by \citet{magicpig}.

%% file: tikz-diagrams/flat-k-vs-quality.tex
%
%

\begin{tikzpicture}[
  meanpoint/.style={draw=white, line width=0.45pt, inner sep=0pt},
  softmaxmark/.style={circle, minimum size=3.2pt},
  topkmark/.style={rectangle, minimum width=3.2pt, minimum height=3.2pt},
  samplingmark/.style={rectangle,rotate=45,
                       minimum width=3.3pt,minimum height=3.3pt},
]
  \definecolor{flatSoftmax}{RGB}{0,114,178}
  \definecolor{flatTopK}{RGB}{213,94,0}
  \definecolor{flatSampling}{RGB}{0,158,115}


  \begin{scope}[x=0.78cm,y=5.175cm]
    \foreach \x/\lab in {0/32,1/64,2/128,3/256,4/512,5/{1,024},6/{2,048}} {
      \draw[black!9,line width=0.35pt] (\x,0) -- (\x,0.8);
      \draw[black!65,line width=0.45pt] (\x,0) -- (\x,-0.0127);
      \node[font=\scriptsize,anchor=north] at (\x,-0.0272) {\lab};
    }
    \foreach \y/\lab in {0/{0},0.1/{0.1},0.2/{0.2},0.3/{0.3},0.4/{0.4},0.5/{0.5},0.6/{0.6},0.7/{0.7},0.8/{0.8}} {
      \draw[black!9,line width=0.35pt] (0,\y) -- (6,\y);
      \draw[black!65,line width=0.45pt] (0,\y) -- (-0.06,\y);
      \node[font=\scriptsize,anchor=east] at (-0.12,\y) {\lab};
    }
    \draw[black!70,line width=0.65pt] (0,0) -- (6,0);
    \draw[black!70,line width=0.65pt] (0,0) -- (0,0.8);
    \node[font=\small,anchor=north] at (3,-0.1127)
      {total retrieval budget $k$};
    \node[font=\small,rotate=90,anchor=south] at (-1.02,0.4)
      {relative $L_2$ error};

    \draw[flatSoftmax,line width=1.05pt]
      (0,0.178147390262) --
      (1,0.119774879600) --
      (2,0.079404578358) --
      (3,0.051554327703) --
      (4,0.032604990375) --
      (5,0.019948450975) --
      (6,0.011729237968);
    \draw[flatTopK,dashed,line width=1.05pt]
      (0,0.497897384574) --
      (1,0.441718878745) --
      (2,0.383098802437) --
      (3,0.322891401709) --
      (4,0.262953135413) --
      (5,0.205561424916) --
      (6,0.152759084801);
    \draw[flatSampling,dash pattern=on 4pt off 1.5pt on 0.8pt off 1.5pt,
          line width=1.05pt]
      (0,0.758941279200) --
      (1,0.648061764888) --
      (2,0.511100396475) --
      (3,0.348298655751) --
      (4,0.198309454049) --
      (5,0.101478826986) --
      (6,0.048828609317);

    \foreach \x/\err in {0/0.178147390262,1/0.119774879600,
      2/0.079404578358,3/0.051554327703,4/0.032604990375,
      5/0.019948450975,6/0.011729237968} {
      \node[meanpoint,softmaxmark,minimum size=4.8pt,fill=flatSoftmax]
        at (\x,\err) {};
    }
    \foreach \x/\err in {0/0.497897384574,1/0.441718878745,
      2/0.383098802437,3/0.322891401709,
      4/0.262953135413,5/0.205561424916,6/0.152759084801} {
      \node[meanpoint,topkmark,minimum width=4.7pt,minimum height=4.7pt,
            fill=flatTopK] at (\x,\err) {};
    }
    \foreach \x/\err in {0/0.758941279200,1/0.648061764888,
      2/0.511100396475,3/0.348298655751,
      4/0.198309454049,5/0.101478826986,6/0.048828609317} {
      \node[meanpoint,samplingmark,minimum size=5.4pt,fill=flatSampling]
        at (\x,\err) {};
    }
  \end{scope}

  \begin{scope}[xshift=7.15cm,x=0.78cm,y=0.0394cm]
    \foreach \x/\lab in {0/32,1/64,2/128,3/256,4/512,5/{1,024},6/{2,048}} {
      \draw[black!9,line width=0.35pt] (\x,0) -- (\x,105);
      \draw[black!65,line width=0.45pt] (\x,0) -- (\x,-1.6);
      \node[font=\scriptsize,anchor=north] at (\x,-3.4) {\lab};
    }
    \foreach \y in {0,20,40,60,80,100} {
      \draw[black!9,line width=0.35pt] (0,\y) -- (6,\y);
      \draw[black!65,line width=0.45pt] (0,\y) -- (-0.06,\y);
      \node[font=\scriptsize,anchor=east] at (-0.12,\y) {\y};
    }
    \draw[black!70,line width=0.65pt] (0,0) -- (6,0);
    \draw[black!70,line width=0.65pt] (0,0) -- (0,105);
    \node[font=\small,anchor=north] at (3,-14.0)
      {total retrieval budget $k$};
    \node[font=\small,rotate=90,anchor=south] at (-1.02,52.5)
      {RULER score};

    \draw[black!55,densely dotted,line width=0.85pt] (0,61.5278) -- (6,61.5278);
    \node[black!60,font=\scriptsize,anchor=south west] at (0.12,63.5)
      {exact attention};
    \draw[flatSoftmax,line width=1.05pt]
      (0,57.5833) -- (1,59.0000) -- (2,60.2222) -- (3,60.9722) --
      (4,61.1944) -- (5,61.8611) -- (6,61.8611);
    \draw[flatTopK,dashed,line width=1.05pt]
      (0,33.1111) -- (1,35.5833) -- (2,40.5556) -- (3,43.3611) --
      (4,41.7778) -- (5,46.3611) -- (6,53.8611);
    \draw[flatSampling,dash pattern=on 4pt off 1.5pt on 0.8pt off 1.5pt,
          line width=1.05pt]
      (0,15.6667) -- (1,21.7222) -- (2,37.1667) -- (3,45.7500) --
      (4,56.6111) -- (5,57.7778) -- (6,61.5278);

    \foreach \x/\score in {0/57.5833,1/59.0000,2/60.2222,3/60.9722,
      4/61.1944,5/61.8611,6/61.8611} {
      \node[meanpoint,softmaxmark,minimum size=4.8pt,fill=flatSoftmax]
        at (\x,\score) {};
    }
    \foreach \x/\score in {0/33.1111,1/35.5833,2/40.5556,3/43.3611,
      4/41.7778,5/46.3611,6/53.8611} {
      \node[meanpoint,topkmark,minimum width=4.7pt,minimum height=4.7pt,
            fill=flatTopK] at (\x,\score) {};
    }
    \foreach \x/\score in {0/15.6667,1/21.7222,2/37.1667,3/45.7500,
      4/56.6111,5/57.7778,6/61.5278} {
      \node[meanpoint,samplingmark,minimum size=5.4pt,fill=flatSampling]
        at (\x,\score) {};
    }
  \end{scope}

  \node[font=\small\bfseries,anchor=south] at (2.34,4.25)
    {(a) Attention error};
  \node[font=\small\bfseries,anchor=south] at (9.49,4.25)
    {(b) Benchmark score};

  \filldraw[fill=black!1,draw=black!15,line width=0.4pt,rounded corners=1.5pt]
    (1.00,-1.80) rectangle (10.70,-1.23);

  \draw[flatSoftmax,line width=1.05pt] (1.38,-1.515) -- (1.94,-1.515);
  \node[meanpoint,softmaxmark,minimum size=4.8pt,fill=flatSoftmax]
    at (1.66,-1.515) {};
  \node[font=\scriptsize,anchor=west] at (2.10,-1.515) {\lift};

  \draw[flatTopK,dashed,line width=1.05pt] (4.42,-1.515) -- (4.98,-1.515);
  \node[meanpoint,topkmark,minimum width=4.7pt,minimum height=4.7pt,
        fill=flatTopK] at (4.70,-1.515) {};
  \node[font=\scriptsize,anchor=west] at (5.14,-1.515) {\topk};

  \draw[flatSampling,dash pattern=on 4pt off 1.5pt on 0.8pt off 1.5pt,
        line width=1.05pt] (6.92,-1.515) -- (7.48,-1.515);
  \node[meanpoint,samplingmark,minimum size=5.4pt,fill=flatSampling]
    at (7.20,-1.515) {};
  \node[font=\scriptsize,anchor=west] at (7.64,-1.515)
    {\topksampling};
\end{tikzpicture}

%% file: sections/algorithm.tex
\section{Algorithms and Technical Overview}
\label{sec:tech-overview}

We give an overview of the main algorithmic ideas in our paper. Our algorithms all share the same primitive of \emph{priority sampling} and differ only in how they retrieve high-priority keys using MIPS.

\paragraph{A meta-algorithm via priority sampling (Appendix~\ref{subsec:meta-alg}).} Fix a query $q$ and let $w_i = \kappa(q,x_i)$ denote the kernel \emph{weight} of key $i$. To approximate attention, we must estimate both $\sum_i w_iv_i$ and $\sum w_i$; simply retrieving the largest weights may not give an accurate estimate. In \emph{priority sampling}, the trick is to assign each key a random scale $\bu_i\sim\Unif(0,1]$ and define the priority $\bpi_i := w_i / \bu_i$. Given the top $K+1$ priorities, the $(K+1)$-st priority acts as a threshold to reweight the first $K$ items (see Figure~\ref{fig:priority-sampling}). The resulting estimator has variance $O(1/K)$ relative to the total weight, and setting $K=\Theta(\eps^{-2})$ yields an $\eps$-accurate estimate. Algorithm~\ref{alg:meta-alg} formalizes this reduction, leaving one algorithmic question: a MIPS oracle retrieves the highest weights $w_i$, but how can we use it to retrieve the highest $K+1$ priorities $\bpi_i$? Our algorithms differ only in how they answer this question. 

\begin{algorithmbox}[label = {alg:meta-alg}]{Meta-Algorithm for Attention via Priority Sampling} 

\textbf{Input.} Keys $x_1,\ldots,x_n\in\R^d$, values $v_1,\ldots,v_n\in\R^D$, and a query $q\in\R^d$ \\
\textbf{Output.} An estimate of $\attn(q)$ \\

\medskip
\textbf{Preprocessing.}
Independently sample $\bu_i\sim\Unif(0,1]$ for each $i\in[n]$

\medskip
\textbf{Query.} Given $q\in\R^d$, \begin{enumerate}
    \item Define $w_i = \kappa(q,x_i)$ and $\bpi_i = w_i / \bu_i$.
    \item \textbf{Priority Retrieval.} Retrieve the $K+1$ indices of highest priority, \[\bpi_{i_1} \ge \bpi_{i_2} \ge \dots \ge \bpi_{i_{K+1}}.\] Let $\bP = \{i_1,\ldots,i_K\}$ and $\btau = \bpi_{i_{K+1}}$.
    \item For each $i\in\bP$, define the adjusted weight $\wh{\bw}_i = \max\{w_i, \btau\}$.
    
    Compute \[\wh{\bV} = \sum_{i\in\bP} v_i\wh{\bw}_i, \qquad \wh{\bW} = \sum_{i\in \bP} \wh{\bw}_i. \]
    \item Return $\wh{\attn}(q) = \wh{\bV} / \wh{\bW}$.
\end{enumerate}
\end{algorithmbox}

\begin{figure}[t]
  \centering
  \resizebox{\linewidth}{!}{\input{tikz-diagrams/priority-sampling}}
  \caption{Illustration of priority sampling for sum estimation.}
  \label{fig:priority-sampling}
\end{figure}
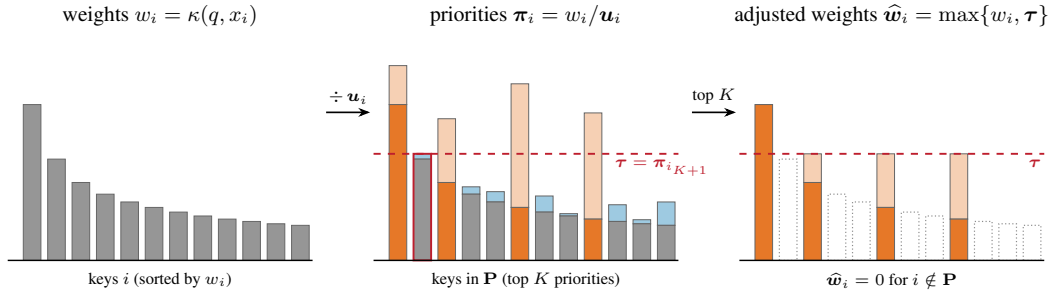

\paragraph{Priority bucketing (Theorems~\ref{thm:bucket-adaptive} and~\ref{thm:bucket-nonadaptive}; Appendix~\ref{subsec:adaptive} and~\ref{subsec:nonadaptive}).} We partition the keys into $T+1$ buckets according to $\bu_i$: for $j=0,\ldots,{T-1}$, bucket $j$ contains keys with $\bu_i\in (\beta^{-(j+1)}, \beta^{-j}]$ for $\beta = (n/(K+1))^{1/(T+1)}$; the final bucket contains all remaining keys. We build MIPS indices on the first $T$ buckets, and since priorities differ by at most a factor $\beta$ within a bucket, a high-weight key is a good candidate for having high priority. The \emph{adaptive} algorithm retrieves the highest-weight key from each bucket and uses the last retrieved key to upper bound the priority of every unseen key in that bucket. It then repeatedly doubles the number of retrieved keys for the bucket with the largest upper bound, stopping once no unseen key can exceed the current $(K+1)$-st largest priority. We prove that at most $O(T + \beta K)$ keys are retrieved in expectation, giving the tradeoff stated in the theorem when $K=\Theta(1/\eps^2)$. The \emph{nonadaptive} variant instead retrieves $\Theta(\beta K)$ keys from each bucket, leading to a factor $T$ increase in retrieval cost.

\paragraph{Lower bounds (Theorems~\ref{thm:lower-bound-const} and~\ref{thm:lower-bound-any}; Appendix~\ref{sec:lower-bound}).} When $T=1$, the above algorithms use a single MIPS index and retrieve $O(\sqrt{n}/\eps)$ keys; our first lower bound shows that this tradeoff is tight. We reduce from a signed sum estimation problem, where $v_i\in\{-1,1\}$, weights $p_i\ge 0$ sum to $1$, and the goal is to estimate $\sum_i v_ip_i$ within additive error $\eps$, given top-$k$ oracle access to the weights. The hard instance hides two random bits: one in whether a single heavy coordinate $p_{\bi^\star}$ is given positive or negative sign, and another in a slight sign bias among $\Theta(\sqrt{n}/\eps)$ ``signal'' coordinates. We also introduce $\Theta(\sqrt{n}/\eps)$ decoy coordinates $p_j$ whose weights are larger than the signal weights. If the indexed set $S$ is small, the algorithm is unlikely to identify the heavy coordinate and first bit; if $S$ is large, then a top-$o(\sqrt{n}/\eps)$ query only returns decoy coordinates, masking the second bit. Thus, any single-index algorithm must retrieve $\Omega(\sqrt{n}/\eps)$ keys. Our second lower bound applies to arbitrary $T$: we show that it is always necessary to retrieve $\smash{\Omt(\log n + 1/\eps^2)}$ keys. This nearly matches our adaptive upper bound of $O(\log n + 1/\eps^2)$ when $T=\Theta(\log n)$.

\paragraph{The softmax lift (Theorem~\ref{thm:softmax-lift}; Appendix~\ref{subsec:softmax-lift}).} We bypass the $\Omega(\sqrt{n}/\eps)$ lower bound for $T=1$ by extending the model slightly and augmenting the keys (and queries) before building the MIPS index. Specifically, we encode the priority directly into inner products by setting $\wt x_i = (x_i,-\log\bu_i)$ and $\wt q = (q,1)$ so that $\exp(\ip{\wt q}{\wt x_i}) = \bpi_i$. Thus, a single top-$(K+1)$ MIPS query returns exactly the top $K+1$ priorities, giving an $\eps$-accurate estimate with only $O(1/\eps^2)$ retrieved keys.

\begin{algorithmbox}[label = {alg:softmax-lift-alg}]{\lift Implementation of Priority Sampling} 

\textbf{Input.} $K \in \N$, $\bu_1, \ldots, \bu_n \sim \Unif(0, 1]$, and $x_1, \ldots, x_n \in \R^d$ \\
\textbf{Output.} Indices $i_1, \ldots, i_{K+1} \in [n]$ of the largest $K+1$ priority keys \\

\textbf{Preprocessing.} 
\begin{enumerate}
    \item Create lifted keys $\tilde{x}_1, \ldots, \tilde{x}_n \in \R^{d+1}$ where
    \[
    \tilde{x}_i := (x_i, -\log \bu_i)
    \]
    and call $\tilde{P} := \{\tilde{x}_1, \ldots, \tilde{x}_n \}$.
    \item $\build(\tilde{P})$
\end{enumerate}

\medskip
\textbf{Query.} Given $q\in\R^d$
\begin{enumerate}
    \item Create lifted query $\tilde{q} := (q, 1) \in \R^{d+1}$
    \item Return $\query(\tilde{q}, K+1; \tilde{P})$
\end{enumerate}
\end{algorithmbox}

%% file: tikz-diagrams/priority-sampling.tex
\begin{tikzpicture}[x=1cm, y=1.15cm]
  \def\psdata{1/2.00/0.80/1, 2/1.30/0.95/2, 3/1.00/0.55/1, 4/0.85/0.90/0,
              5/0.75/0.85/0, 6/0.68/0.30/1, 7/0.62/0.75/0, 8/0.57/0.95/0,
              9/0.53/0.28/1, 10/0.50/0.70/0, 11/0.47/0.90/0, 12/0.45/0.60/0}
  \def\pstau{1.368}
  \def\pswidth{0.26}   
  \def\psstep{0.36}    
  \def\pspanel{4.5}    

  \begin{scope}
    \foreach \idx/\wv/\uv/\fl in \psdata {
      \pgfmathsetmacro\xleft{\idx*\psstep-\pswidth/2}
      \fill[psW, psbar] (\xleft,0) rectangle ++(\pswidth,\wv);
    }
    \draw[line width=0.6pt] (0,0) -- (\pspanel,0);
    \node[pstitle] at (\pspanel/2,2.9) {weights $w_i = \kappa(q,x_i)$};
    \node[font=\scriptsize, anchor=north] at (\pspanel/2,-0.05) {keys $i$ (sorted by $w_i$)};
  \end{scope}

  \begin{scope}[xshift=5.4cm]
    \foreach \idx/\wv/\uv/\fl in \psdata {
      \pgfmathsetmacro\xleft{\idx*\psstep-\pswidth/2}
      \pgfmathsetmacro\pri{\wv/\uv}
      \ifnum\fl=1
        \fill[psS!35, psbar] (\xleft,\wv) rectangle ++(\pswidth,\pri-\wv);
        \fill[psS, psbar]    (\xleft,0)   rectangle ++(\pswidth,\wv);
      \else
        \fill[psP, psbar] (\xleft,\wv) rectangle ++(\pswidth,\pri-\wv);
        \fill[psW, psbar] (\xleft,0)   rectangle ++(\pswidth,\wv);
      \fi
      \ifnum\fl=2
        \draw[psT, line width=0.8pt] (\xleft,0) rectangle ++(\pswidth,\pri);
      \fi
    }
    \draw[pstauline] (0,\pstau) -- (\pspanel,\pstau);
    \node[psT, font=\scriptsize, anchor=north east, inner sep=1pt]
      at (\pspanel+0.5,\pstau-0.04) {$\btau = \bpi_{i_{K+1}}$};    \draw[line width=0.6pt] (0,0) -- (\pspanel,0);
    \node[pstitle] at (\pspanel/2,2.9) {priorities $\bpi_i = w_i/\bu_i$};
    \node[font=\scriptsize, anchor=north] at (\pspanel/2,-0.05) {keys in $\bP$ (top $K$ priorities)};
  \end{scope}

  \begin{scope}[xshift=10.8cm]
    \foreach \idx/\wv/\uv/\fl in \psdata {
      \pgfmathsetmacro\xleft{\idx*\psstep-\pswidth/2}
      \ifnum\fl=1
        \pgfmathsetmacro\hgt{max(\wv,\pstau)}
        \fill[psS!35, psbar] (\xleft,\wv) rectangle ++(\pswidth,\hgt-\wv);
        \fill[psS, psbar]    (\xleft,0)   rectangle ++(\pswidth,\wv);
      \else
        \draw[black!45, densely dotted, line width=0.5pt]
          (\xleft,0) rectangle ++(\pswidth,\wv);
      \fi
    }
    \draw[pstauline] (0,\pstau) -- (\pspanel,\pstau);
    \node[psT, font=\scriptsize, anchor=north east, inner sep=1pt]
      at (\pspanel,\pstau-0.04) {$\btau$};
    \draw[line width=0.6pt] (0,0) -- (\pspanel,0);
    \node[pstitle] at (\pspanel/2,2.9) {adjusted weights $\wh{\bw}_i = \max\{w_i,\btau\}$};
    \node[font=\scriptsize, anchor=north] at (\pspanel/2,-0.05) {$\wh{\bw}_i = 0$ for $i \notin \bP$};
  \end{scope}

  \draw[-{Stealth[length=5pt]}, line width=0.7pt]
    (4.7,1.9) -- node[above, font=\scriptsize]{$\div\,\bu_i$} (5.35,1.9);
  \draw[-{Stealth[length=5pt]}, line width=0.7pt]
    (10.1,1.9) -- node[above, font=\scriptsize]{top $K$} (10.75,1.9);
\end{tikzpicture}

%% file: sections/experiments.tex
\section{Experimental Results} \label{sec:experiments}

We investigate the empirical performance of our retrieval strategies.

\textbf{Hardware.} We ran all experiments on a machine equipped with an AMD Ryzen Threadripper 9970X CPU (32 physical cores, 64 hardware threads), 256 GB of system RAM, and an NVIDIA RTX PRO 6000 Blackwell Max-Q GPU with 96 GB of VRAM.

\paragraph{Experimental Procedure.} We replace the attention computation of Llama-3.2-1B-Instruct with a custom harness which allows the user to specify an arbitrary MIPS oracle and an arbitrary retrieval strategy that uses the oracle as a black-box.

Each example is divided into one main document and a complex question at the end. In the preprocessing stage, the document prefill is done using the unmodified model, and oracles are allowed to build MIPS indexes on the resulting KV-cache. As is standard practice~\citep{saap}, we prune the recent context and sink from the KV-cache ($1024$ tokens and $1$ token, respectively) and compute attention on this segment via brute-force. At query time, the question is provided and the modified LLM implementing our approximate attention algorithm is used to generate a response.

For each strategy/oracle combination, the harness measures benchmark scores, attention error and final logit error relative to exact attention, oracle recall rates, question processing time, and the per-token latency of the response. These values are measured for each head on each token where applicable, but all plots report aggregate statistics. 

Note that in all algorithms, $k$ is used to denote the total number of points requested from all the oracle calls combined. The precise definitions of the measured values, oracle configurations, and additional experimental discussion is given in Appendix \ref{sec:exp-apx}.

\paragraph{Dataset.} The dataset examples come from a variety of tasks from RULER~\citep{hsieh2024ruler}, a synthetic long-context benchmark. We chose RULER because it supports very large contexts (128K--1M) and its tasks can be solved by small models such as Llama-3.2-1B-Instruct.

\paragraph{Variable $k$ results.}We test the accuracy of the retrieval strategies as a function of $k$.  Figure \ref{fig:flat-k-vs-quality} shows attention error and task accuracy as a function of $k$ for the FAISS-Flat oracle.  Note that \lift (Algorithm \ref{alg:softmax-lift-alg}) and \adaptivebuckets (Algorithm \ref{alg:adaptive-alg}) are different implementations of the same Meta Algorithm \ref{alg:meta-alg}, so their output is identical when using an exact oracle. Figure \ref{fig:ivf-latency-vs-benchmark} shows a latency-quality tradeoff for retrieval strategies using the FAISS-IVF oracle instead. The latency of \topksampling is larger than that of \topk, even though we found the raw oracle times to be roughly comparable for the two strategies. This suggests that the overhead from the unoptimized merge with the sampling results is causing the discrepancy.  

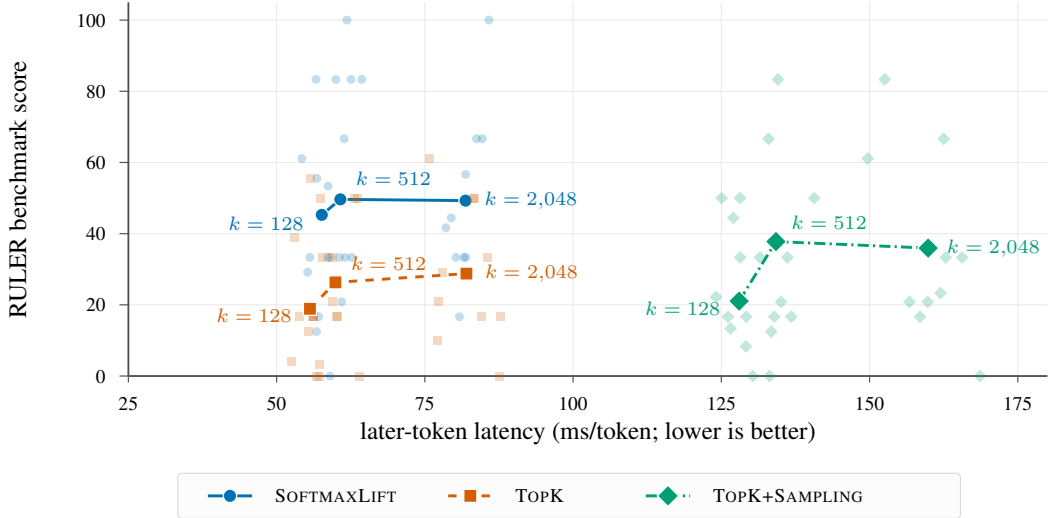
\begin{figure}[t]
  \centering
  \resizebox{\linewidth}{!}{%
    \input{tikz-diagrams/ivf-latency-vs-benchmark}%
  }
  \caption{Latency--quality tradeoff for strategies using the FAISS-IVF oracle on 60 examples. Each translucent point averages one of 10 tasks; opaque points denote overall averages. For Top-$k$+sampling, the budget is split evenly between oracle
  retrieval and uniform samples.}
  \label{fig:ivf-latency-vs-benchmark}
\end{figure}

  \paragraph{Variable-oracle results.} 
  Table~\ref{tab:retrieval-oracle-results} compares the 4 retrieval strategies across 3 FAISS oracles (Flat, IVF, HNSW)
  at a fixed total budget of $k=128$. As above, \lift outperforms the other retrieval strategies. Note that the slight difference in scores between the exact attention methods is due to differences in numerical precision. 
  At this context size, the unoptimized approximate methods on the CPU are still slower than the model's original exact attention run on the GPU. 
  
  \begin{table*}[t]
    \centering
    \small
    \setlength{\tabcolsep}{6pt}
    \renewcommand{\arraystretch}{1.20}
    \begin{tabular}{@{}llrrrrr@{}}
      \toprule
      \multirow{2}{*}{\textbf{Retrieval strategy}}
        & \multirow{2}{*}{\textbf{Oracle}}
        & \textbf{Attention} & \textbf{RULER} & \textbf{Response}
        & \textbf{Oracle} \\
        & & {\footnotesize rel. $L_2$ error $\downarrow$}
        & {\footnotesize score (\%) $\uparrow$}
        & {\footnotesize time (s) $\downarrow$}
        & {\footnotesize recall $\uparrow$} \\
      \midrule
      \multirow{3}{*}{\lift}
        & Flat & 0.0794 & 60.22 & 5.10 & 1 \\
        & IVF  & \textbf{0.1209} & 45.28 & 1.15 & 0.966  \\
        & HNSW & 0.1212 & \textbf{50.53} & \textbf{1.06} & 0.926\\
      \addlinespace
      \multirow{3}{*}{\topk}
        & Flat & 0.3831 & 40.56 & 5.39 & 1  \\
        & IVF  & 0.4009 & 18.89 & 1.40 & 0.821  \\
        & HNSW & 0.3935 & 23.42 & 1.54 & 0.918  \\
      \addlinespace
      \multirow{3}{*}{\topksampling}
        & Flat & 0.5111 & 37.17 & 9.32 & 1  \\
        & IVF  & 0.5567 & 21.06 & 5.19 & 0.820 \\
        & HNSW & 0.5259 & 24.44 & 5.40 & 0.938 \\
      \addlinespace
      \multirow{3}{*}{\adaptivebuckets}
        & Flat
               & 0.0794 & 60.22 & 6.33 & 1  \\
        & IVF  & 0.1404 & 37.53 & 2.77 & 0.780  \\
        & HNSW & 0.0948 & 57.86 & 3.81 & 0.983 \\
      \midrule
      Exact CPU        & -- & -- & 60.86 & 3.37 & --  \\
      Exact GPU        & -- & 0 & 61.53 & 0.53 & -- \\
      \bottomrule
    \end{tabular}
    \caption{Retrieval strategy and oracle comparison on 60 matched 128K-context
    RULER examples from 10 tasks with total retrieval budget
    $k=128$. \topksampling splits its budget evenly between 64 retrieved and 64 uniformly
    sampled keys. We use $T=2$ for \adaptivebuckets in this table. We bold the best performing retrieval strategy/approximate oracle combination (among those with response time faster than Exact CPU) along each metric.}
    \label{tab:retrieval-oracle-results}
  \end{table*}

\paragraph{1M Context Results.} Large contexts favor sparse attention methods due to the sublinear scaling of MIPS oracles and memory hardware limitations for exact algorithms (see Appendix \ref{sec:exp-apx}). We test our approach on a small set of 1M context examples. For these experiments, we modify Llama-3.1-Nemotron-8B-UltraLong-1M-Instruct instead and use the same testing protocol. \lift remains the best retrieval strategy. Notably, even though our code is not heavily optimized, \lift with an IVF oracle run on the CPU can generate answer tokens at a roughly 3x larger rate than the exact attention done on the GPU while still succeeding on some tasks.

\begin{table*}[t]
    \centering
    \small
    \setlength{\tabcolsep}{8pt}
    \renewcommand{\arraystretch}{1.20}
 \begin{tabular}{@{}llrrr@{}}
    \toprule
    \multirow{2}{*}{\textbf{Retrieval strategy}}
      & \multirow{2}{*}{\textbf{Oracle}}
      & \textbf{Attention} & \textbf{RULER} & \textbf{Later-token} \\
      & & {\footnotesize rel. $L_2$ error $\downarrow$}
      & {\footnotesize score $\uparrow$}
      & {\footnotesize latency (s/token) $\downarrow$} \\
    \midrule
    \multirow{2}{*}{\lift}
      & Flat & 0.0181 & 83.33 & 2.567 \\
      & IVF  & \textbf{0.0670} & \textbf{58.33} & 0.519 \\
    \addlinespace
    \multirow{2}{*}{\topk}
      & Flat & 0.2516 & 75.00 & 2.563 \\
      & IVF  & 0.2975 & 37.50 & \textbf{0.409} \\
    \addlinespace
    \multirow{2}{*}{\topksampling}
      & Flat & 0.1823 & 72.92 & 4.460 \\
      & IVF  & 0.2946 & 45.83 & 0.796 \\
    \midrule
    Exact GPU & -- & 0 & 85.42 & 1.695 \\
    \bottomrule
  \end{tabular}
    \caption{1M-context results on eight matched examples from four RULER
    tasks. Approximate methods use an 8192-token exact local window and total
    nonlocal budget $k=1024$. Attention error is pooled relative $L_2$;
    later-token latency averages tokens generated after the first. We bold the best performing retrieval strategy/approximate oracle combination along each metric.}
    \label{tab:one-million-results}
  \end{table*}

\paragraph{Discussion.} Our experiments demonstrate that \lift consistently provides the best quality--efficiency tradeoff among the tested retrieval strategies. While none of the approximate methods outperform the exact GPU attention at 128K context, the 1M context experiments suggest that sparse attention is advantageous at longer context lengths. In particular, combining \lift with an approximate oracle appears to be a promising overall candidate, making a strong case for further optimization and engineering. 

%% file: tikz-diagrams/ivf-latency-vs-benchmark.tex
%

\begin{tikzpicture}[
  x=0.075cm,
  y=0.045cm,
  taskpoint/.style={draw=none, opacity=0.24, inner sep=0pt},
  meanpoint/.style={draw=white, line width=0.45pt, inner sep=0pt},
  softmaxmark/.style={circle, minimum size=3.2pt},
  topkmark/.style={rectangle, minimum width=3.2pt, minimum height=3.2pt},
  samplingmark/.style={rectangle,rotate=45,
                       minimum width=3.3pt,minimum height=3.3pt},
]
  \definecolor{latSoftmax}{RGB}{0,114,178}
  \definecolor{latTopK}{RGB}{213,94,0}
  \definecolor{latSampling}{RGB}{0,158,115}

  \foreach \x in {25,50,75,100,125,150,175} {
    \draw[black!9, line width=0.35pt] (\x,0) -- (\x,105);
    \draw[black!65, line width=0.45pt] (\x,0) -- (\x,-1.8);
    \node[font=\scriptsize, anchor=north] at (\x,-3.0) {\x};
  }
  \foreach \y in {0,20,40,60,80,100} {
    \draw[black!9, line width=0.35pt] (25,\y) -- (180,\y);
    \draw[black!65, line width=0.45pt] (25,\y) -- (23.5,\y);
    \node[font=\scriptsize, anchor=east] at (22.8,\y) {\y};
  }
  \draw[black!70, line width=0.65pt] (25,0) -- (180,0);
  \draw[black!70, line width=0.65pt] (25,0) -- (25,105);
  \node[font=\small, anchor=north] at (102.5,-10.0)
    {later-token latency (ms/token; lower is better)};
  \node[font=\small, rotate=90, anchor=south] at (9.5,52.5)
    {RULER benchmark score};

  \foreach \x/\y in {
    54.2551/61.1111, 61.3876/66.6667, 60.5666/33.3333,
    56.7284/12.5000, 55.2398/29.1667, 56.6492/83.3333,
    59.9966/83.3333, 58.6206/33.3333, 57.0799/16.6667,
    55.6074/33.3333
  } \node[taskpoint,softmaxmark,fill=latSoftmax] at (\x,\y) {};
  \foreach \x/\y in {
    56.7933/55.5556, 64.3610/83.3333, 61.5909/33.3333,
    60.9771/20.8333, 59.0228/33.3333, 61.8585/100.0000,
    62.5437/83.3333, 62.6900/33.3333, 59.0201/0.0000,
    58.6674/53.3333
  } \node[taskpoint,softmaxmark,fill=latSoftmax] at (\x,\y) {};
  \foreach \x/\y in {
    79.4557/44.4444, 84.6871/66.6667, 81.9043/33.3333,
    80.2330/33.3333, 78.5565/41.6667, 85.8007/100.0000,
    83.7129/66.6667, 81.5619/33.3333, 80.8308/16.6667,
    81.8725/56.6667
  } \node[taskpoint,softmaxmark,fill=latSoftmax] at (\x,\y) {};

  \draw[latSoftmax, line width=1.05pt]
    (57.6131,45.2778) -- (60.7525,49.6389) -- (81.8615,49.2778);
  \node[meanpoint,softmaxmark,minimum size=4.8pt,fill=latSoftmax]
    at (57.6131,45.2778) {};
  \node[meanpoint,softmaxmark,minimum size=4.8pt,fill=latSoftmax]
    at (60.7525,49.6389) {};
  \node[meanpoint,softmaxmark,minimum size=4.8pt,fill=latSoftmax]
    at (81.8615,49.2778) {};
  \node[latSoftmax,font=\scriptsize,anchor=east] at (56.2,43.0) {$k=128$};
  \node[latSoftmax,font=\scriptsize,anchor=south west] at (61.7,51.0) {$k=512$};
  \node[latSoftmax,font=\scriptsize,anchor=west] at (83.4,49.3) {$k=2{,}048$};

  \foreach \x/\y in {
    53.0857/38.8889, 57.6948/33.3333, 56.0087/16.6667,
    55.4064/12.5000, 52.4453/4.1667, 57.2339/0.0000,
    56.7872/0.0000, 57.3913/50.0000, 53.9062/16.6667,
    56.2556/16.6667
  } \node[taskpoint,topkmark,fill=latTopK] at (\x,\y) {};
  \foreach \x/\y in {
    55.6687/55.5556, 63.6767/50.0000, 60.0523/16.6667,
    59.4161/20.8333, 56.3016/16.6667, 63.9589/0.0000,
    63.1629/50.0000, 59.3769/33.3333, 60.2702/16.6667,
    57.2850/3.3333
  } \node[taskpoint,topkmark,fill=latTopK] at (\x,\y) {};
  \foreach \x/\y in {
    75.8114/61.1111, 83.3710/50.0000, 87.7181/16.6667,
    77.2737/20.8333, 77.9809/29.1667, 87.5150/0.0000,
    83.2464/50.0000, 85.5544/33.3333, 84.5658/16.6667,
    77.1049/10.0000
  } \node[taskpoint,topkmark,fill=latTopK] at (\x,\y) {};

  \draw[latTopK,dashed,line width=1.05pt]
    (55.6215,18.8889) -- (59.9169,26.3056) -- (82.0142,28.7778);
  \node[meanpoint,topkmark,minimum width=4.7pt,minimum height=4.7pt,
        fill=latTopK]
    at (55.6215,18.8889) {};
  \node[meanpoint,topkmark,minimum width=4.7pt,minimum height=4.7pt,
        fill=latTopK] at (59.9169,26.3056) {};
  \node[meanpoint,topkmark,minimum width=4.7pt,minimum height=4.7pt,
        fill=latTopK]
    at (82.0142,28.7778) {};
  \node[latTopK,font=\scriptsize,anchor=east] at (54.2,17.0) {$k=128$};
  \node[latTopK,font=\scriptsize,anchor=south west] at (60.9,27.0) {$k=512$};
  \node[latTopK,font=\scriptsize,anchor=west] at (83.5,28.8) {$k=2{,}048$};

  \foreach \x/\y in {
    124.1155/22.2222, 128.1636/50.0000, 130.3111/0.0000,
    129.1481/8.3333, 126.1421/16.6667, 133.1535/0.0000,
    128.1954/33.3333, 125.0678/50.0000, 129.2070/16.6667,
    126.5531/13.3333
  } \node[taskpoint,samplingmark,fill=latSampling] at (\x,\y) {};
  \foreach \x/\y in {
    127.0316/44.4444, 140.6668/50.0000, 136.7978/16.6667,
    135.0815/20.8333, 133.4160/12.5000, 136.1465/33.3333,
    134.5552/83.3333, 132.9678/66.6667, 131.5177/33.3333,
    133.9336/16.6667
  } \node[taskpoint,samplingmark,fill=latSampling] at (\x,\y) {};
  \foreach \x/\y in {
    149.6765/61.1111, 162.5105/66.6667, 162.8314/33.3333,
    156.7302/20.8333, 159.8024/20.8333, 168.6407/0.0000,
    152.5590/83.3333, 165.6060/33.3333, 158.5100/16.6667,
    161.9636/23.3333
  } \node[taskpoint,samplingmark,fill=latSampling] at (\x,\y) {};

  \draw[latSampling,dash pattern=on 4pt off 1.5pt on 0.8pt off 1.5pt,
        line width=1.05pt]
    (128.0057,21.0556) -- (134.2115,37.7778) -- (159.8830,35.9444);
  \node[meanpoint,samplingmark,minimum size=5.4pt,fill=latSampling]
    at (128.0057,21.0556) {};
  \node[meanpoint,samplingmark,minimum size=5.4pt,fill=latSampling]
    at (134.2115,37.7778) {};
  \node[meanpoint,samplingmark,minimum size=5.4pt,fill=latSampling]
    at (159.8830,35.9444) {};
  \node[latSampling,font=\scriptsize,anchor=east] at (126.5,19.0) {$k=128$};
  \node[latSampling,font=\scriptsize,anchor=south west] at (135.2,38.5) {$k=512$};
  \node[latSampling,font=\scriptsize,anchor=west] at (161.4,35.9) {$k=2{,}048$};

  \begin{scope}[x=1cm,y=1cm]
    \filldraw[fill=black!1,draw=black!15,line width=0.4pt,
              rounded corners=1.5pt]
      (2.50,-1.80) rectangle (12.40,-1.23);

    \draw[latSoftmax,line width=1.05pt] (2.88,-1.515) -- (3.44,-1.515);
    \node[meanpoint,softmaxmark,minimum size=4.8pt,fill=latSoftmax]
      at (3.16,-1.515) {};
    \node[font=\scriptsize,anchor=west] at (3.60,-1.515) {\lift};

    \draw[latTopK,dashed,line width=1.05pt] (5.92,-1.515) -- (6.48,-1.515);
    \node[meanpoint,topkmark,minimum width=4.7pt,minimum height=4.7pt,
          fill=latTopK] at (6.20,-1.515) {};
    \node[font=\scriptsize,anchor=west] at (6.64,-1.515) {\topk};

    \draw[latSampling,
          dash pattern=on 4pt off 1.5pt on 0.8pt off 1.5pt,
          line width=1.05pt] (8.42,-1.515) -- (8.98,-1.515);
    \node[meanpoint,samplingmark,minimum size=5.4pt,fill=latSampling]
      at (8.70,-1.515) {};
    \node[font=\scriptsize,anchor=west] at (9.14,-1.515)
      {\topksampling};
  \end{scope}
\end{tikzpicture}

%% file: sections/limitations.tex
\section{Limitations and Future Work}\label{sec:limitations}

We discuss important limitations of our work, both from the theoretical and empirical lenses. We hope to address many of these limitations as part of future work. 

\paragraph{Theoretical Limitations.} Our theoretical model does not completely capture the complexity of implementing search oracles, rather opting to simply charge per used oracle. Furthermore, the oracle cost $\costo$ in our model does not distinguish between oracles built over differing size supports. Secondly, our theoretical results assume access to \emph{exact} search oracles. While our empirical results indicate that they are very robust to \emph{approximate} search oracles, we leave open the theoretical problem of designing algorithms provably robust to approximate search oracles. 

\paragraph{Experimental Limitations.} Our MIPS Attention model importantly excludes some existing work that (1) do not treat the vector search as black-box~\citep{magicpig}, or (2) allow arbitrary transformations to the keys and queries~\citep{hashattention}. We leave the additional empirical comparison of our algorithms with these as future work. Secondly, our empirical testing suite does not fully optimize vector search oracles via multi-threading and does not support GPU-accelerated search. Lastly, our experiments at 1M context use too few examples to be fully conclusive. 

\subsection*{AI use and Acknowledgments}

In this work, we used generative AI tools to assist in the construction and analysis of all algorithms and theorems, in writing the proofs, and in writing code for our implementations. We have not used generative AI tools for generating synthetic data sets, developing theoretical models or conceptual frameworks, designing research methodology, or interpreting results. We take responsibility for the entire content of the paper.

This work was supported in part by the Columbia Center of AI Technology in collaboration with Amazon, and ONR (N00014-22-1-2713); by the National Science Foundation under Grant No. CCF-2337993; and by the National Science Foundation GRFP under Grant No. DGE-2236662.

%% file: sections/algorithms-appendix.tex
\section{Algorithms for MIPS Attention}\label{sec:algorithms}

\subsection{A Meta-Algorithm via Priority Sampling}\label{subsec:meta-alg}

We refer to Algorithm~\ref{alg:meta-alg}. Fixing a query $q$, we define
\[
    w_i := \kappa(q,x_i), \qquad
    V := \sum_{i=1}^n v_i w_i,
    \qquad
    W := \sum_{i=1}^n w_i,
\]

So that $\attn(q) = V/W$. We also define $\wt{\bw}_i = \wh{\bw}_i \cdot \one\{i\in\bP\}$, so that:
\[
    \widehat{\bV}=\sum_{i=1}^n v_i\wt \bw_i,
    \qquad
    \widehat{\bW}=\sum_{i=1}^n \wt \bw_i.
\]

\begin{lemma}[Fact~1 and Theorem~2 of~\cite{daliri2024simple}]
\label{lem:priority-sampling}
    $\E[\wt\bw_i] = w_i$ for every $i\in[n]$ and $\Cov(\wt\bw_i, \wt\bw_j) = 0$ for every $i\neq j$. Finally, $\sum_{i=1}^n \Var(\wt\bw_i) \le W^2 / (K-1)$.
\end{lemma}

Using Lemma~\ref{lem:priority-sampling}, we bound the error of the attention numerator and denominator.
\begin{claim}
\label{claim:num-denom}
For every $\delta\in(0,1)$, with probability at least $1-\delta$,
\[
    |\widehat{\bW}-W|
    \le
    W\sqrt{\frac{2}{\delta(K-1)}} \quad \text{and} \quad \|\widehat{\bV}-V\|_2
    \le
    \max_{i\in[n]} \|v_i\|_2 \cdot W\sqrt{\frac{2}{\delta(K-1)}}.
\]
\end{claim}

\begin{proof} By Lemma~\ref{lem:priority-sampling}, we have $
    \E[(\wt\bw_i-\bw_i)(\wt\bw_j-\bw_j)] = \Cov(\wt\bw_i,\wt\bw_j) = 0$.
Therefore, \[\E[|\wh\bW - W|^2] =  \E\left[\left|\sum_{i=1}^n (\wt\bw_i-\bw_i)\right|^2\right] = \sum_{i=1}^n \Var(\wt\bw_i) \le W^2 / (K-1).
\]
Similarly,
\[\E\|\widehat{\bV}-V\|_2^2 = \E\left\|
        \sum_{i=1}^n v_i(\wt\bw_i-\bw_i)
    \right\|_2^2  = \sum_{i=1}^n \|v_i\|_2^2\Var(\wt\bw_i) \le \max_{i\in[n]} \|v_i\|_2 \cdot W^2 / (K-1).\] The desired statement follows from applying Chebyshev's inequality with failure probability $\delta/2$ for both $\wh\bW$ and $\wh\bV$, and then taking a union bound.
\end{proof}

\begin{lemma}
\label{lem:attn-err}
For every $\delta\in(0,1)$, with probability at least $1-\delta$,
\[
    \|\widehat{\attn}(q)-\attn(q)\|_2
    =
    O\!\left(\frac{1}{\sqrt{\delta K}}\right) \cdot \max_{i\in[n]}\|v_i\|_2.
\]
\end{lemma}

\begin{proof}
    By Claim~\ref{claim:num-denom}, with probability $1-\delta$ we have $|\wh\bW-W|\le \eta W$ and $\|\wh\bV-V\|_2 \le \eta W\cdot \max_{i\in[n]} \|v_i\|_2$ for 
\[
    \eta
    :=
    \sqrt{\frac{2}{\delta(K-1)}}.
\]

The desired bound is vacuously true when $K = O(1/\delta)$, so we may assume that $K$ is large enough for $\eta\le 1/2$.
Then,\begin{align*}
    \|\widehat{\attn}(q)-\attn(q)\|_2
    &=
    \left\|
        \frac{\widehat{\bV}}{\widehat{\bW}}
        -
        \frac{V}{W}
    \right\|_2 \\
    &\le
    \frac{\|\widehat{\bV}-V\|_2}{\widehat{\bW}}
    +
    \frac{\|V\|_2\cdot
          |\widehat{\bW}-W|}
         {W\cdot\widehat{\bW}} \\
    &\le
    \left(\frac{\eta}{1-\eta}
    +
    \frac{\eta}{1-\eta}\right)\cdot \max_{i\in[n]} \|v_i\|_2 \\
    &=
    \frac{2\eta}{1-\eta} \cdot \max_{i\in[n]} \|v_i\|_2
    \le 4\eta \cdot \max_{i\in[n]} \|v_i\|_2,
\end{align*}
completing the proof.
\end{proof}

\subsection{Adaptive Algorithm}\label{subsec:adaptive}

In this section, we present one instantiation of the Meta Algorithm \ref{alg:meta-alg}. Algorithm \ref{alg:adaptive-alg} adaptively queries its oracles to implement the \textbf{priority retrieval} step.

\begin{algorithmbox}[label = {alg:adaptive-alg}]{Adaptive Bucket Implementation of Priority Sampling}

\textbf{Input.} $K \in \N$, $T \in \N$, $\bu_1, \ldots, \bu_n \sim \Unif(0, 1]$, and $x_1, \ldots, x_n \in \R^d$ \\
\textbf{Output.} Indices $i_1, \ldots, i_{K+1} \in [n]$ of the largest $K+1$ priority keys \\

\textbf{Preprocessing.} If $n \leq K+1$, the task is trivial. Otherwise, 
\begin{enumerate}
    \item Set $\beta=\left(\frac{n}{K+1}\right)^{1/(T+1)}$
    \item Create buckets
\begin{align*}
    P_j&=\left\{i \in [n]:\beta^{-(j+1)}
        <\bu_i\leq\beta^{-j}\right\},
        &&j=0,\ldots,T-1,\\
        P_{T}&=\left\{i \in [n]:\bu_i
        \leq\beta^{-T}\right\}.
\end{align*}
\item $\build(P_j)$ for $j = 0, \ldots, T-1$.
\end{enumerate}

\medskip
\textbf{Query.} Given $q\in\R^d$, let $k_j = 1$ for $j \in 0, \ldots, T-1$. 
\begin{enumerate}
    \item $\query(q, k_j; P_j)$ for $j = 0, \ldots, T-1$.
    \item Compute $\bpi_i$ for each $i \in P_T$ via direct inner product evaluations

    \item Define the \emph{frontier} of bucket $P_j$ to be 
    \[
    F_j := \beta^{j+1} \kappa(x^{j}_{\downarrow})
    \]
    where $x^{j}_{\downarrow}$ is the lowest rank point retrieved from bucket $j$ thus far. Set $F_j = 0$ if all the points in bucket $P_j$ have been read. 
    \item Let $\tau$ be the $(K+1)$st largest positive priority $\bpi_i$ seen so far
    (or zero if fewer than $K+1$ points have been seen).

    \item While $\tau<\max_{j= 0, \ldots , T-1} F_j$, 
    
    \begin{enumerate}
        \item Let $j^{\star}:= \argmax_{j = 0, \ldots, T-1} F_j$ and set $k_{j^{\star}} := 2\cdot k_{j^{\star}}$.
        \item $\query(q, k_{j^{\star}}; P_{j^{\star}})$
        \item Update $F_{j^{\star}}$ and $\tau$ according to new retrieved points
    \end{enumerate}

    \item Return the $K+1$ indices with the largest priorities seen
\end{enumerate}
\end{algorithmbox}

\begin{lemma}\label{lem:adaptive-alg-retrieved} For any query $q \in \R^d$ and $T \in \N$, Algorithm \ref{alg:adaptive-alg} returns the $K+1$ indices with the largest priorities, has $\costo = T$, and 
\[
    \E[\costr] \leq O\!\left(T+n^{1/(T+1)}K^{1 - 1/(T+1)} \right)
\]
\end{lemma}
\begin{proof}
    Observe that the frontier $F_j$ of a bucket $P_j$ represents the largest possible priority of the remaining unretrieved keys from bucket $P_j$. The termination condition of the while loop in step 5 of Algorithm \ref{alg:adaptive-alg} ensures that while there exists a bucket that could potentially contain a key of priority larger than the $(K+1)$st largest we've seen thus far, we continue to retrieve keys. Thus, by step 6 of the algorithm, the $(K+1)$ keys with the largest priorities seen coincide exactly with the global $K+1$ keys of highest priority. Algorithm \ref{alg:adaptive-alg} builds exactly $T$ oracles during preprocessing, therefore $\costo = T$.      
    
    It remains to upper bound the expected retrieval cost. We let $\btau^{\star}$ denote the true $(K+1)$-st largest priority $\bpi_i$ among the $n$ keys. Let $\calC = \{i\in[n] : \bpi_i \ge \btau^\star / \beta\}$. We claim that $\E[|\calC|] \le \beta(K+1)$. To see this, let $\btau_{-i}$ denote the $(K+1)$-st largest priority in $[n]\setminus\{i\}$. Note that $i$ is a top $K+1$ priority $\iff$ $\bpi_i \ge \btau_{-i}$, while $i\in\calC \iff \bpi_i \ge \btau_{-i} / \beta$.
    
    We can compute the following quantities, where $\bpi_{-i}$ denotes the collection of all priorities other than item $i$: \[\Pr[\bpi_i \ge \btau_{-i} / \beta \mid \bpi_{-i}] = \min\{1, \beta w_i / \btau_{-i}\},\quad \Pr[\bpi_i \ge \btau_{-i} \mid \bpi_{-i}] = \min\{1, w_i / \btau_{-i}\}.\] Hence, \[\Pr[i\in\calC \mid \bpi_{-i}] \le \beta \cdot \Pr[i\text{ is a top $(K+1)$ priority} \mid \bpi_{-i}].\] Summing over $i$, \[\E[|\calC|] \le \beta \cdot \sum_{i} \Pr[i\text{ is a top $(K+1)$ priority}] = \beta(K+1).\]

    We now show that we can effectively charge every bucket's retrieval cost to $\calC$. Consider a bucket $j$ the last time the algorithm doubles its query size, and let $k_j$ be the query size before doubling. Since $j$ is selected, it must be that $F_j = \max_{\ell} F_\ell \ge \btau^\star$, or else the algorithm would have terminated. Then, each of the $k_j$ keys $i$ retrieved so far must have $\bpi_i \ge \beta^j \kappa(x^j_{\downarrow}) = F_j / \beta \ge \btau^\star/\beta$, implying $k_j \le |\calC \cap P_j|$. Therefore, bucket $j$ contributes at most $1+2+\dots+2k_j \le 4k_j\le 4|\calC\cap P_j|$ to the retrieval cost. Summing over disjoint buckets and adding the $|P_T|$ direct evaluations gives $\costr \le T + 4|\calC| + |P_T|$. Since $\E[|P_T|] = n\beta^{-T} = (K+1)\beta$, we get $\E[\costr] \le T + 5(K+1)\beta$. Plugging in the setting of $\beta$ gives the final claim. 
\end{proof}

We can now prove Theorem \ref{thm:bucket-adaptive} by combining the guarantees of Lemmas \ref{lem:attn-err} and \ref{lem:adaptive-alg-retrieved}. 

\begin{proof}[Proof of Theorem \ref{thm:bucket-adaptive}]
    We instantiate the Meta-Algorithm \ref{alg:meta-alg} with $K = \Theta(1/\epsilon^2)$ and use Algorithm \ref{alg:adaptive-alg} to implement the \textbf{Priority Retrieval} step. By Lemma \ref{lem:adaptive-alg-retrieved}, Algorithm \ref{alg:adaptive-alg} returns the $(K+1)$ keys of largest priority and retrieves $O(T + n^{1/(T+1)}K^{1-1/(T+1)})$ keys in expectation. By Markov's inequality, Algorithm \ref{alg:adaptive-alg} retrieves at most $10$ times as many keys with probability at least $0.9$. Taking $\delta = 0.1$, Lemma \ref{lem:attn-err} then allows us to conclude that the Meta-Algorithm \ref{alg:meta-alg}, using Algorithm \ref{alg:adaptive-alg} to implement the priority retrieval step, returns an approximation 
    \begin{align*}
        \norm{\widehat{\attn}(q) - \attn(q)}_2 \leq \epsilon \cdot \max_{i \in [n]} \norm{v_i}_2
    \end{align*}
    using $T$ oracles and retrieving $O(T + n^{1/(T+1)}\epsilon^{-2 + 2/(T+1)})$ keys.\footnote{While we have bounded the expected retrieval cost, we can easily extend this to worst-case retrieval cost by early-stopping the algorithm with a small loss in failure probability.} The algorithm succeeds with probability at least $2/3$. 
\end{proof}

\subsection{Non-Adaptive Algorithm}\label{subsec:nonadaptive}

\begin{algorithmbox}[label = {alg:non-adaptive-alg}]{Non-Adaptive Bucket Implementation of Priority Sampling} 

\textbf{Input.} $K \in \N$, $T \in \N$, $\bu_1, \ldots, \bu_n \sim \Unif(0, 1]$, and $x_1, \ldots, x_n \in \R^d$ \\
\textbf{Output.} Indices $i_1, \ldots, i_{K+1} \in [n]$ of the largest $K+1$ priority keys \\

\textbf{Preprocessing.} If $n \leq K+1$, return $x_1, \ldots, x_n$. Otherwise, 
\begin{enumerate}
    \item Set $\beta=\left(\frac{n}{K+1}\right)^{1/(T+1)}$.
    \item Create buckets
\begin{align*}
    P_j&=\left\{i \in [n]:\beta^{-(j+1)}
        <\bu_i\leq\beta^{-j}\right\},
        &&j=0,\ldots,T-1,\\
        P_{T}&=\left\{i \in [n]:\bu_i
        \leq\beta^{-T}\right\}.
\end{align*}
\item $\build(P_j)$ for $j = 0, \ldots, T-1$.
\end{enumerate}

\medskip
\textbf{Query.} Given $q\in\R^d$, set $r := \lceil C(K+1) \beta\rceil$ for sufficiently large constant $C$. If $|P_{T}| > r$, report \texttt{Fail}. Otherwise,
\begin{enumerate}
    \item $\query(q, r; P_j)$ for $j = 0, \ldots, T-1$
    \item Compute $\bpi_i$ for each $i \in P_T$ via direct inner product evaluations
    \item Return the $K+1$ indices with the largest priorities seen
\end{enumerate}
\end{algorithmbox}

\begin{lemma}\label{lem:non-adaptive-alg-retrieved} For any query $q \in \R^d$ and $T \in \N$, Algorithm \ref{alg:non-adaptive-alg} returns the $K+1$ indices with the largest priorities, has $\costo = T$, and 
\[
    \costr \leq O\!\left(
        T n^{1/(T+1)} K^{1 - 1/(T+1)} \right)
\]
\end{lemma}
\begin{proof}
    We let $\btau^{\star}$ denote the true $(K+1)$-st largest priority $\bpi_i$ among the $n$ keys. Let $\calC = \{i\in[n] : \bpi_i \ge \btau^\star / \beta\}$. By the exact same reasoning as in the proof of Lemma~\ref{lem:adaptive-alg-retrieved}, we have $\E[|\calC|] \le \beta(K+1)$.

    We let $\calE$ denote the event that both $|\calC| \leq r$ and $|P_T| \leq r$. Observe that by Markov's inequality and an union bound, $\Pr[\ol{\calE}] \leq 2/C$. Conditioning on event $\calE$, Algorithm \ref{alg:non-adaptive-alg} does not report $\texttt{Fail}$ and the priorities of all points in bucket $P_T$ are computed via direct inner product evaluations. 
    
    Suppose, for the sake of contradiction that there exists an index $i \in [n]$ such that $\bpi_i$ is among the $(K+1)$ largest, $i \in P_j$ for some $j \in \{0, \ldots, T-1\}$, and index $i$ is \emph{not} retrieved in step 1 of the query algorithm. Let $S$ denote the $r$ indices retrieved from bucket $j$ in step 1. Then, for each $\ell \in S$, it must be the case that $\kappa_q(\ell) \geq \kappa_q(i)$. Furthermore, it must be the case that
    \begin{align*}
        \bpi_{\ell} = \frac{\kappa_q(\ell)}{\bu_{\ell}} &\geq \frac{\kappa_q(i)}{\bu_{\ell}} \geq \frac{\kappa_q(i)}{\beta \bu_i} = \bpi_i / \beta \geq \tau^{\star} / \beta
    \end{align*}
    But then, along with index $i$, there would be at least $r+1$ indices in $\calC$, a contradiction to event $\calE$. Thus, assuming event $\calE$, the top $(K+1)$ priority indices are returned by Algorithm \ref{alg:non-adaptive-alg}. 
    
    Choosing $C = 20$, gives that event $\calE$ occurs with probability $\geq 0.9$. Exactly $T$ oracles are built during preprocessing. Conditioned on event $\calE$, $Tr$ keys are retrieved and at most $r$ direct inner product evaluations are performed. Thus, with probability at least $0.9$, $\costo = T$ and $\costr = O( T n^{1/(T+1)} K^{1 - 1/(T+1)})$. 
\end{proof}

We can now prove Theorem \ref{thm:bucket-nonadaptive} by combining the guarantees of Lemmas \ref{lem:attn-err} and \ref{lem:non-adaptive-alg-retrieved}. 

\begin{proof}[Proof of Theorem \ref{thm:bucket-nonadaptive}]
    We instantiate the Meta-Algorithm \ref{alg:meta-alg} with $K = \Theta(1/\epsilon^2)$ and use Algorithm \ref{alg:non-adaptive-alg} to implement the \textbf{Priority Retrieval} step. By Lemma \ref{lem:non-adaptive-alg-retrieved}, Algorithm \ref{alg:non-adaptive-alg} returns the $(K+1)$ keys of largest priority and has retrieval cost $O( T n^{1/(T+1)}K^{1-1/(T+1)})$ with probability at least $0.9$. Taking $\delta = 0.1$, Lemma \ref{lem:attn-err} then allows us to conclude that the Meta-Algorithm \ref{alg:meta-alg}, using Algorithm \ref{alg:non-adaptive-alg} to implement the priority retrieval step, returns an approximation 
    \begin{align*}
        \norm{\widehat{\attn}(q) - \attn(q)}_2 \leq \epsilon \cdot \max_{i \in [n]} \norm{v_i}_2
    \end{align*}
    using $T$ oracles and retrieving $O(T n^{1/(T+1)}\epsilon^{-2 + 2/(T+1)})$ keys. The algorithm succeeds with probability at least $2/3$. 
\end{proof}

\subsection{Softmax Lift Algorithm}\label{subsec:softmax-lift}

\begin{lemma}\label{lem:softmax-lift-alg-retrieved} For any query $q \in \R^d$, Algorithm \ref{alg:softmax-lift-alg} returns the $K+1$ indices with the largest priorities, has $\costo = 1$, and
\[
    \costr = O(\min\{ n, K \})
\]
\end{lemma}
\begin{proof}
     Observe that for a given query $q \in \R^d$ and for every $i \in [n]$, 
     \begin{align*}
         \langle \tilde{q}, \tilde{x}_i \rangle = \langle q, x_i \rangle - \log(\bu_i) = \log \left( \frac{\exp( \langle q, x \rangle)}{\bu_i} \right) = \log \bpi_i
     \end{align*}
     Since the natural logarithm is an increasing function, the indices of the augmented keys which have the largest inner products with $\tilde{q}$ exactly coincide with the keys of largest priority. Observe that only one oracle is built during preprocessing and trivially, $\min\{n, K+1\}$ keys are retrieved. Thus, $\costo = 1$ and $\costr = O( \min \{n, K\})$.
\end{proof}

We can now prove Theorem \ref{thm:softmax-lift} by combining the guarantees of Lemmas \ref{lem:attn-err} and \ref{lem:softmax-lift-alg-retrieved}. 

\begin{proof}[Proof of Theorem \ref{thm:softmax-lift}]
    We instantiate the Meta-Algorithm \ref{alg:meta-alg} with $K = \Theta(1/\epsilon^2)$ and use Algorithm \ref{alg:softmax-lift-alg} to implement the \textbf{Priority Retrieval} step. By Lemma \ref{lem:softmax-lift-alg-retrieved}, Algorithm \ref{alg:softmax-lift-alg} returns the $(K+1)$ keys of largest priority and retrieves $O(\min\{n, K\})$ keys. Taking $\delta = 0.1$, Lemma \ref{lem:attn-err} then allows us to conclude that the Meta-Algorithm \ref{alg:meta-alg}, using Algorithm \ref{alg:softmax-lift-alg} to implement the priority retrieval step, returns an approximation 
    \begin{align*}
        \norm{\widehat{\attn}(q) - \attn(q)}_2 \leq \epsilon \cdot \max_{i \in [n]} \norm{v_i}_2
    \end{align*}
    using $1$ oracle and retrieving $O(1/\epsilon^2)$ keys. The algorithm succeeds with probability at least $0.9$. 
\end{proof}

%% file: sections/lower-bound.tex
\newcommand{\topquery}{\mathsf{Top}}
\newcommand{\eval}{\mathsf{Eval}}

\section{Lower Bounds for MIPS Attention}\label{sec:lower-bound}

We prove lower bounds matching the upper bounds established above. We begin by introducing \emph{signed sum estimation}, a simpler problem for which all our lower bounds will be proven.

\subsection{Signed Sum Estimation}

\begin{definition}[Signed sum estimation]\label{def:signed-sum}
    In this problem, an algorithm is given signs $v_1,\ldots,v_n\in\{-1,1\}$. It may construct up to $T$ subsets $S_1,\ldots,S_T\subseteq[n]$. There are unknown weights $p_1,\ldots,p_n\ge 0$ satisfying $\sum_{i=1}^n p_i=1$, which the algorithm can only access using the following operations:
    \begin{itemize}
        \item $\topquery(k;S_t)$ returns the indices and weights of the $\min\{k,|S_t|\}$ largest weights in $S_t$, breaking ties lexicographically, at retrieval cost $k$.
        \item $\eval(i)$ returns $p_i$ at retrieval cost $1$.
    \end{itemize}
    Its retrieval cost is the sum of the individual costs of the operations, and its oracle cost is $T$. The goal is to estimate $\mu=\sum_{i=1}^n v_ip_i$.

    An algorithm is $\eps$-accurate if, for every fixed choice of signs and weights, it outputs an estimate $\widehat\mu$ satisfying
    \[
        \Pr[|\widehat\mu-\mu|\le\eps]\ge\frac23,
    \]
    where the probability is over the algorithm's randomness.
\end{definition}

The following reduction allows us to establish lower bounds for attention by considering only signed sum estimation.

\begin{lemma}\label{lem:sum-to-attn}
    Any $\eps$-accurate attention algorithm with oracle budget $T$ yields an $\eps$-accurate signed sum estimation algorithm with the same oracle cost and retrieval cost.
\end{lemma}

\begin{proof}
    Let $\calA$ be an $\eps$-accurate attention algorithm. Given signs $v_1,\ldots,v_n$, construct an attention instance with keys $x_i=e_i\in\R^n$ and the same (scalar) values $v_i$. Run the preprocessing stage of $\calA$ and select the corresponding subsets $S_1,\ldots,S_T$.

    Let $p_1,\ldots,p_n\ge0$ be the unknown weights, and consider the query $q\in\R^n$ given by $q_i=\tau\log p_i$, where $\log0=-\infty$.\footnote{To avoid setting infinite values, we may apply a slight perturbation to all zero weights and renormalize; for simplicity of presentation, we omit this detail.} We simulate $\calA$ on this query. Since $\ip{q}{x_i}=\log p_i$, we can answer $\query(q,k;S_t)$ with $\topquery(k;S_t)$, and we can answer direct inner product evaluations with $\eval(i)$. These operations have identical retrieval costs. Finally, we output $\widehat\mu=\widehat{\attn}(q)$. It follows that $|\widehat\mu-\mu|\le\eps$ with the required probability, since
    \[
        \attn(q)
        =\frac{\sum_{i=1}^n\exp(q_i)v_i}{\sum_{i=1}^n\exp(q_i)}
        =\sum_{i=1}^n p_iv_i=\mu,
    \]
    and $\max_i|v_i|=1$.
\end{proof}

In particular, Lemma~\ref{lem:sum-to-attn} implies that lower bounds for signed sum estimation extend directly to lower bounds for attention. The following standard distribution testing lower bound will be useful in proving the lower bounds.

\begin{lemma}[Theorem~2.2 and Lemma~2.5 of~\cite{tsybakov2009}] \label{lem:coin-test-lb}
    For $\delta\in(0,1/2)$, let $\bB\sim\{-1,1\}$, and $\bZ_1,\ldots,\bZ_\ell\sim\{-1,1\}$ be i.i.d. with $\Pr[\bZ_i = \bB] = \frac{1+\delta}{2}$. For any estimate $\wh{\bB} = \wh{\bB}(\bZ_1,\ldots,\bZ_\ell)$, \[\Pr[\wh{\bB} = \bB] \le 1/2 + O(\delta\sqrt{\ell}).\]
\end{lemma}

\subsection{Lower Bound for $T=1$}

By Lemma~\ref{lem:sum-to-attn}, it suffices to prove the following statement.

\begin{theorem}\label{thm:single-oracle-sum-lb}
    For sufficiently small $\eps>0$, any $\eps$-accurate signed sum estimation algorithm with oracle cost at most $1$ has retrieval cost $\Omega(\min\{n, \sqrt{n} / \eps\})$.
\end{theorem}

By Yao's minimax principle, it suffices to construct a distribution over instances on which no deterministic algorithm with small retrieval cost can output an $\eps$-accurate estimate with probability $\ge 2/3$.

\paragraph{The hard distribution.} For convenience, we assume that $n$ is even. Let $N=n/2$ and index the coordinates in $[n]$ by $[N]\times\{-1,1\}$. Set values $v_{j,b} = b \text{ for all } (j,b)\in [N]\times \{-1,1\}$. Now, fixing a sufficiently small constant $\alpha>0$, define \[m = \floor{\alpha \min\{n, \sqrt{n}/\eps\}},\quad\delta=\max\{32\eps, 16/\sqrt{m}\}\in (0,1/2).\] We sample uniformly random disjoint sets $\bA_{sig},\bA_{dec}\subseteq[N]$ with size $|\bA_{sig}| = |\bA_{dec}| = m$, and a uniformly random $\bj^{\star}\sim [N] \setminus (\bA_{sig}\cup \bA_{dec})$.

We independently sample $\bB_0, \bB_1\sim\{-1,1\}$ uniformly at random. The idea will be to encode these bits in the signed sum, such that both bits must be recovered to achieve an $\eps$-accurate approximation. Specifically, we select weights as follows: \begin{itemize}
    \item \textbf{Heavy pair:} set $w_{\bj^\star, \bB_0} = 1$.
    \item \textbf{Signal pairs:} for each $j\in\bA_{sig}$, independently sample $\bZ_j\in\{-1,1\}$ with $\Pr[\bZ_j = \bB_1] = (1+\delta)/2$. Then, set $w_{j,\bZ_j} = 1/m$.
    \item \textbf{Decoy pairs:} for each $j\in\bA_{dec}$, set $w_{j,-1} = w_{j,1} = 2/m$.
    \item All other weights are set to $0$.
\end{itemize}

The total sum of weights is $W = \sum_{j,b} w_{j,b} = 6$. Let $p_{j,b} = w_{j,b}/W$ be the normalized weights, and $\mu = \sum_{j,b} v_{j,b} p_{j,b} = \frac{1}{6}(\bB_0 + \frac{1}{m}\sum_{j\in\bA_{sig}} \bZ_j)$ be the target signed sum.

\begin{claim} \label{claim:est-to-decode}
    Any algorithm that outputs an $\eps$-accurate estimate $\wh{\mu}$ with probability $2/3$ can also output $\wh{\bB_0},\wh{\bB_1}$ satisfying $(\wh{\bB_0},\wh{\bB_1}) = (\bB_0,\bB_1)$ with probability $> 0.66$.
\end{claim}
\begin{proof}
    Define the event \[\calG = \left\{\left|\frac{1}{m}\sum_{j\in\bA_{sig}} \bZ_j-\delta\bB_1\right| \le \delta/3\right\}.\] By Hoeffding's inequality, $\Pr[\calG] \ge 1-2e^{-m\delta^2/18} > 0.999$, since $\delta\ge 16/\sqrt{m}$.

    Assuming $\calG$, we have $|\mu - \frac{1}{6}(\bB_0 + \delta\bB_1)|\le \delta/18$ and thus \[|\wh{\mu} - \frac{1}{6}(\bB_0 + \delta\bB_1)|\le \delta/18 + \delta/32 < \delta/10.\] Since $\delta<5/8$, the algorithm can recover $\bB_0 = \sign(\wh{\mu})$ and $\bB_1 = \sign(\wh{\mu}-\bB_0 / 6)$ with probability $2/3-0.001\ge0.66$.
\end{proof}

Next, we will show that it is impossible for an algorithm with oracle cost $1$ and retrieval cost $K=o(\min\{n, \sqrt{n}/\eps\})$ to output $(\bB_0,\bB_1)$ with probability substantially more than $1/2$. Let $S,\bQ\subseteq[N]\times\{-1,1\}$ denote the set chosen by the algorithm and the (adaptively chosen) indices that it queries directly. We divide into cases according to the size of $S$. It will be convenient to define: \begin{itemize}
    \item $J(S) = \{j\in[N] : \{(j,-1),(j,1)\}\cap S \ne \emptyset\}$, the set of pairs touched by $S$.
    \item $J(\bQ) = \{j\in[N] : \{(j,-1),(j,1)\}\cap \bQ \ne \emptyset\}$, the set of pairs touched by $\bQ$.
\end{itemize}

\begin{claim}\label{claim:small-set-case}
    If $|S| < n/20$, then the probability the algorithm recovers $\bB_0$ is at most $0.6$.
\end{claim}
\begin{proof}
    In order to recover $\bB_0$ with probability $\ge 1/2$, it must be that $\bj^\star\in J(S)\cup J(\bQ)$. We compute $|J(\bS) + J(\bQ)|/(N-2m) \le (n/20 + K)/(N-2m) \le 1/5$. Since $\bj^\star$ is chosen uniformly at random among $N-2m$ indices, the probability of failure is at least $\frac12 \cdot \frac45 = \frac25$.
\end{proof}

\begin{claim}\label{claim:large-set-case}
    If $|S|\ge n/20$, then the probability the algorithm recovers $\bB_1$ is at most $0.6$.
\end{claim}
\begin{proof}
    Lemma~\ref{lem:coin-test-lb} implies that if the number of pairs in $\bA_{sig}$ observed by the algorithm is $\ell = o(1/\eps^2)$, then it can recover $\bB_1$ only with probability $\le 1/2+o(1)$. We will first show that, with high probability, it observes $0$ weights in $\bA_{sig}$ from $\topquery(k; S)$.

    Let $\bh = |\bA_{dec}\cap J(S)|$ denote the number of decoy pairs intersecting $S$. Note that $\E[\bh] \ge |S|\cdot (m/n) \ge m/20$. Defining the event $\calE = \{\bh \ge m/40\}$, and noting that $\bh$ follows a hypergeometric distribution, a standard lower tail bound gives $\Pr[\calE] \ge 1-e^{-m/160}\ge 0.95$. 
    
    Assuming $\calE$ and using $k\le K = o(m) \le m/40$, it follows that $\topquery(k; S)$ gives no pairs in $\bA_{sig}$.

    The only other weights observed by the algorithm are the at-most $K$ direct evaluations. In expectation, only $O(Km/n) = o(1/\eps^2)$ of the chosen pairs will be in $\bA_{sig}$. By Markov's inequality, only $o(1/\eps^2)$ are observed with probability $0.95$. 
    
    Taking a union bound, we conclude that with probability $0.9$, the algorithm observes $o(1/\eps^2)$ pairs in $\bA_{sig}$, and so it fails with probability at least $(1/2-o(1))\cdot 0.9 > 0.4$, as desired.
\end{proof}

\begin{proof}[Proof of Theorem~\ref{thm:single-oracle-sum-lb}]
    By Claim~\ref{claim:est-to-decode}, any $\eps$-accurate algorithm must identify $(\bB_0,\bB_1)$ with probability $>0.66$. However, by Claims~\ref{claim:small-set-case} and~\ref{claim:large-set-case}, no algorithm with oracle cost $1$ and retrieval cost $K=o(\sqrt{n}/\eps)$ can recover $(\bB_0,\bB_1)$ with probability more than $0.6$. We conclude that any $\eps$-accurate signed sum estimation algorithm with oracle cost $1$ has retrieval cost $\Omega(\min\{n,\sqrt{n}/\eps\})$, as desired.
\end{proof}

\subsection{Lower Bound for Arbitrary $T$}

Next, we show that regardless of $T$, any $\eps$-accurate attention algorithm must retrieve $\Omega(\log n / \log\log n + 1/\eps^2)$ keys. By Lemma~\ref{lem:sum-to-attn}, it suffices to show the following statement.

\begin{theorem}\label{thm:any-oracle-sum-lb}
    For sufficiently small $\eps>0$, any $\eps$-accurate signed sum estimation algorithm has retrieval cost \[\Omega\left(\min\left\{n,\, \frac{\log n}{\log\log n} + \frac{1}{\eps^2}\right\}\right),\] regardless of its oracle cost.
\end{theorem}

To prove this, we will construct two separate instances, one that requires retrieval cost $\Omega(\min\{n, 1/\eps^2\})$ and one that requires retrieval cost $\Omega(\log n / \log\log n)$.

\subsubsection{The $\Omega(1/\eps^2)$ Bound}

\begin{lemma}\label{lem:eps-2-lb}
    Any $\eps$-accurate signed sum estimation algorithm has retrieval cost $\Omega(\min\{n, 1/\eps^2\})$.
\end{lemma}

By Yao's minimax principle, it suffices to construct a distribution over instances on which no deterministic algorithm with small retrieval cost can output an $\eps$-accurate estimate with probability at least $2/3$.

\paragraph{The hard distribution.} For convenience, assume that $n$ is even. Let $N=n/2$ and index elements in $[n]$ by $[N]\times\{-1,1\}$. Set $v_{j,b}=b$ for all $(j,b)\in[N]\times\{-1,1\}$. Sample a random bit $\bB\sim\{-1,1\}$ uniformly at random. Choose weight assignments as follows: \begin{itemize}
    \item For each $j\in[N]$, sample $\bZ_j\sim\{-1,1\}$ such that $\Pr[\bZ_j = \bB] = 1/2+\eps$.
    \item Set $\bw_{j,\bZ_j} = 1$ and $\bw_{j,-\bZ_j} = 0$.
\end{itemize}

Then, $W = \sum_{j,b} \bw_{j,b} = N$. Setting $\bp_{j,b} = \bw_{j,b} / W$, the normalized signed sum is \[\bmu = \frac{1}{N} \cdot \sum_{j=1}^N \bZ_j.\]

\begin{claim} \label{claim:est-reveals-single-bit}
    Any algorithm that outputs an $\eps$-accurate estimate $\wh{\mu}$ with probability $2/3$ can also output $\bB$ with probability $>0.66$.
\end{claim}
\begin{proof}
    Note that $\E[\bmu] = 2\bB\eps$. Let $\calG$ denote the event that $|\bmu - 2\bB\eps| < \eps$. By a Chernoff bound, we have $\Pr[\calG] \ge 1 - 2e^{-N\eps^2/2} \ge 0.999$, since we can assume without loss of generality that $\eps \ge c/\sqrt{n}$ for a sufficiently large constant $c$. Therefore, \[|\wh{\bmu}-2\bB\eps| < 2\eps,\] meaning $\sign(\wh{\bmu}) = \bB$ with probability $2/3-0.001> 0.66$.
\end{proof}

\begin{claim} \label{claim:no-recovering-bit}
    Any algorithm with retrieval cost $o(1/\eps^2)$ recovers $\bB$ with probability at most $0.51$.
\end{claim}
\begin{proof}
    We first show that the algorithm might as well make no $\topquery$ calls, up to a constant factor increase in retrieval cost. To simulate $\topquery(k;S)$, simply scan the pairs intersecting $S$ in lexicographic order, evaluating one unseen coordinate from each pair until $k$ weight $1$ elements have been found. Each fresh pair contributes a weight $1$ element with probability at least $1/2-\eps = \Omega(1)$, so the expected number of direct evaluations needed to simulate $\topquery$ is $O(k)$. It follows that an algorithm of retrieval cost $K$ can be simulated with $O(K)$ direct evaluations in expectation, and therefore with high probability via Markov's inequality.
    
    Assuming the algorithm only makes direct evaluations, Lemma~\ref{lem:coin-test-lb} implies that $\ge \Omega(1/\eps^2)$ retrieved points are necessary to succeed with probability $0.51$, so we conclude that $K\ge\Omega(1/\eps^2)$.
\end{proof}

\begin{proof}[Proof of Lemma~\ref{lem:eps-2-lb}]
    Follows directly from combining Claims~\ref{claim:est-reveals-single-bit} and~\ref{claim:no-recovering-bit}.
\end{proof}

\subsubsection{The $\Omega(\log n / \log\log n)$ Bound}

\begin{lemma}\label{lem:log-log-lb}
For sufficiently small $\eps>0$, any $\eps$-accurate signed sum estimation algorithm has retrieval cost $\Omega(\log n / \log\log n)$.
\end{lemma}

By Yao's minimax principle, we may fix a distribution over instances and rule out deterministic algorithms.

\paragraph{The hard distribution.}
For convenience, assume that $n$ is even, let $N=n/2$, and index elements in $[n]$ by $[N]\times\{-1,1\}$. We set values $v_{j,b}=b$ for each $(j,b)\in[N]\times \{-1,1\}$. We will divide $[N]$ into $L$ levels, where $L = \floor{c\log n / \log\log n}$ for a sufficiently small constant $c>0$. In particular, we will sample uniformly random disjoint sets

$$
\bA_1,\ldots,\bA_L\subseteq[N],\qquad |\bA_\ell| = s_\ell := (4L)^{\ell-1}.
$$

Note that $\sum_{\ell=1}^L s_\ell \le N$ if $c$ is sufficiently small. We will sample a uniformly random special block $\bell^{\star}\sim[L]$ and a random bit $\bB\sim\{-1,1\}$.
\begin{itemize}
\item For every ordinary block $\ell\ne \bell^{\star}$ and $j\in\bA_\ell$, we set
\(        p_{j,-1} = p_{j,1} = \frac{1}{4(L-1)s_\ell}.
   \)
\item For $j\in\bA_{\bell^{\star}}$, set $p_{j,\bB} = 3 / (8s_{\bell^\star})$ and $p_{j,-\bB} = 1 / (8s_{\bell^\star})$.
\item We set all other $p_{j,b} = 0$.
\end{itemize}

Therefore, the sum of $p_{j,b}$ within each ordinary block is $1/(2(L-1))$ and the sum within the special block is $1/2$, confirming that $\sum_{j,b} p_{j,b} = 1$. Meanwhile, the signed sum is

\[
    \bmu
    = \sum_{j,b} b\,p_{j,b}
    = \bB/4
\]

since in the special level, there are $s_{\bell^\star}$ pairs, each contributing $\bB/(4s_{\bell^\star})$.

\begin{claim}\label{claim:log-est-reveals-bit}
For $\eps < 1/4$, any algorithm that outputs an $\eps$-accurate estimate $\wh{\mu}$ can recover $\bB$.
\end{claim}

\begin{proof}
If $|\wh{\mu}-\bmu| < 1/4$, then $\sign(\wh{\mu}) = \sign(\bmu) = \bB$.
\end{proof}

Next, we will use the fact that weights monotonically decrease with level. In particular, for $\ell < \ell'$, any $j\in\bA_\ell$, $j'\in\bA_{\ell'}$, and $b,b'\in\{-1,1\}$ satisfy $p_{j,b}>p_{j',b'}$.

Indeed, the smallest possible weight in level $\ell$ is $\frac{1}{4(L-1)s_\ell}$ and the largest is

$$
    \frac{3}{8s_{\ell'}}
    \le \frac{3}{32Ls_\ell}
    < \frac{1}{4(L-1)s_\ell},
$$

where we used $s_{\ell'} / s_\ell \ge 4L$.

\begin{claim}\label{claim:log-must-hit-special}
Any algorithm with retrieval cost $K$ recovers $\bB$
with probability at most $1/2+K/(2L)$.
\end{claim}
\begin{proof}
Say that the algorithm \emph{touches} level $\ell$ if it retrieves some $(j,b)$ with $j\in\bA_\ell$. Since weights strictly decrease with level, any $\topquery(k;S)$ query retrieves elements of $S$ in increasing order of level. In particular, given that the algorithm has not touched level $\bell^\star$, then the distribution of $\bell^\star$ conditioned on the information received by the algorithm is uniform in the remaining levels.

Since the algorithm touches at most $K$ levels, we conclude that $\bell^\star$ is touched with probability $\le K/L$. Given that $\bell^\star$ is not touched, $\bB$ can only be recovered with probability $1/2$. Therefore, $\bB$ is recovered with probability at most $1/2 \cdot (1-K/L) + K/L = 1/2 + K/(2L)$, as desired.
\end{proof}

\begin{proof}[Proof of Lemma~\ref{lem:log-log-lb}]
    Combining Claims~\ref{claim:log-est-reveals-bit} and~\ref{claim:log-must-hit-special} yields $K = \Omega(L) = \Omega(\log n / \log\log n)$.
\end{proof}

Combining Lemma~\ref{lem:log-log-lb} with Lemma~\ref{lem:eps-2-lb}, every $\eps$-accurate signed sum estimation algorithm has retrieval cost

$$
    \Omega\left(
        \min\left\{
            n,\,
            \frac{\log n}{\log\log n}
            +\frac{1}{\eps^2}
        \right\}
    \right),
$$

which proves Theorem~\ref{thm:any-oracle-sum-lb}.

%% file: sections/experiments-appendix.tex
\section{Further Experimental Details}\label{sec:exp-apx}

\subsection{Measurement Definitions and Details}
\label{sec:measurement-definitions}

For each strategy--oracle pair and benchmark example, the harness runs three
separate passes:

\begin{itemize}
    \item \textbf{Accuracy pass.} We replay the exact-GPU token sequence with
    teacher forcing, ensuring that both executions see the same token prefix.
    The pass records logits, attention contexts, and oracle results.  Each exact
    attention reference is recomputed using the same query and frozen K/V as
    the approximation, isolating local attention error from error propagated
    through earlier layers.

    \item \textbf{Latency pass.} Free greedy decoding is instrumented with
    component timers for index construction, oracle search, CPU attention, and
    data movement.  We use this pass only for build times and timing breakdowns,
    since its instrumentation perturbs end-to-end latency.

    \item \textbf{Clean pass.} Free greedy decoding is repeated without
    accuracy or component-level instrumentation.  This provides the
    generated answer, benchmark score, first-token latency, later-token latency,
    and total response time.
\end{itemize}

We define the following measurements:
\begin{itemize}
    \item \textbf{Logit error.} Let $\mathbf{z}_t$ and $\widehat{\mathbf{z}}_t$ denote the exact and
approximate full-vocabulary logit vectors at output position $t$. Global relative $L_2$ error over the
    teacher-forced output trajectory:
    \begin{equation*}
        \left(
        \frac{\sum_{t}\lVert\widehat{\mathbf{z}}_t-\mathbf{z}_t\rVert_2^2}
             {\sum_{t}\lVert\mathbf{z}_t\rVert_2^2}
        \right)^{1/2}.
    \end{equation*}

    \item \textbf{Attention error.} Let
$\mathbf{c}_r$ and $\widehat{\mathbf{c}}_r$ denote exact and approximate
attention contexts for layer--head--query row $r$. Global relative $L_2$ error over all
    recorded attention contexts:
    \begin{equation*}
        \left(
        \frac{\sum_{r}\lVert\widehat{\mathbf{c}}_r-\mathbf{c}_r\rVert_2^2}
             {\sum_{r}\lVert\mathbf{c}_r\rVert_2^2}
        \right)^{1/2}.
    \end{equation*}

    \item \textbf{Oracle recall.} For oracle invocation $r$, let
    $S_r^\star$ be the exact top-$k$ set and $\widehat{S}_r$ the set returned by
    the oracle.  Recall is
    \begin{equation*}
        \frac{\sum_r |\widehat{S}_r \cap S_r^\star|}
             {\sum_r |S_r^\star|}.
    \end{equation*}

    \item \textbf{First-token latency.} Wall-clock time from the
    start of question processing through selection of the first generated
    token.

    \item \textbf{Later-token latency.} Mean wall-clock time for a
    complete decode step after the first token. 
\end{itemize}

\subsection{SoftmaxLift oracle comparison}

Table~\ref{tab:softmax-lift-all-oracles} expands the main-text oracle
comparison to all six oracles for which we have matched \lift results.
Among the approximate oracles, PiPNN gives the best
quality, while HNSW gives the lowest end-to-end response latency.  MLANN has
the lowest later-token latency, but its lower retrieval accuracy substantially
reduces its RULER score.
\begin{table}[H]
    \centering
    \small
    \setlength{\tabcolsep}{4.5pt}
    \renewcommand{\arraystretch}{1.20}
    \begin{tabular}{@{}lrrrrrrrr@{}}
      \toprule
      \multirow{2}{*}{\textbf{Oracle}}
        & \multicolumn{4}{c}{\textbf{Accuracy}}
        & \multicolumn{4}{c}{\textbf{Latency}} \\
      \cmidrule(lr){2-5}\cmidrule(l){6-9}
        & \textbf{Logit} & \textbf{Attention} & \textbf{RULER}
        & \textbf{Oracle} & \textbf{Build} & \textbf{First token}
        & \textbf{Later token} & \textbf{Response} \\
        & {\footnotesize rel. $L_2$ $\downarrow$}
        & {\footnotesize rel. $L_2$ $\downarrow$}
        & {\footnotesize score $\uparrow$}
        & {\footnotesize recall $\uparrow$}
        & {\footnotesize (s) $\downarrow$}
        & {\footnotesize (ms) $\downarrow$}
        & {\footnotesize (ms/token) $\downarrow$}
        & {\footnotesize (s) $\downarrow$} \\
      \midrule
      Flat  & 0.0849 & 0.0794 & 60.22
            & 1.000 & 6.429 & 3258.1 & 136.7 & 5.096 \\
      IVF   & 0.1820 & 0.1209 & 45.28 & \textbf{0.966}
            & \textbf{29.350} & 397.1 & 55.7 & 1.148 \\
      HNSW  & 0.1670 & 0.1212 & 50.53 & 0.926
            & 441.902 & \textbf{349.7} & 51.8 & \textbf{1.063} \\
      NSG   & \textbf{0.1588} & \textbf{0.1093} & \textbf{52.44} & 0.937
            & 417.301 & 676.2 & 57.1 & 1.456 \\
      MLANN & 0.2500 & 0.1662 & 25.58 & 0.871
            & 100.121 & 498.2 & \textbf{48.7} & 1.137 \\
      PiPNN & 0.1116 & 0.0989 & 58.69 & 0.902
            & 574.074 & 3232.6 & 141.1 & 5.168 \\
      \midrule
      Exact CPU        & -- & -- & 60.86 & -- & -- & -- & -- & 3.37  \\
      Exact GPU        & 0 & 0 & 61.53 & -- & -- & -- & -- & 0.53\\
      \bottomrule
    \end{tabular}
  \caption{\lift on the same 60 matched
  128K-context RULER examples from 10 tasks, using total retrieval budget
  $k=128$.  The exact-GPU mean
  RULER score is 61.53. We bold the best performing approximate oracle (among those with response time faster than Exact CPU) along each metric.}
  \label{tab:softmax-lift-all-oracles}
  \end{table}

\subsection{Oracle Parameters}\label{sec:oracle-parameters}

Table~\ref{tab:oracle-parameters} gives the oracle operating points used in
Table~\ref{tab:softmax-lift-all-oracles}.  These parameters were fixed across
all 60 examples.  Flat, IVF, and MLANN operate directly on inner products;
HNSW, NSG, and PiPNN use the exact Xbox reduction from MIPS to Euclidean
nearest-neighbor search~\citep{bachrach2014xbox}.

\begin{table}[H]
  \centering
  \small
  \setlength{\tabcolsep}{5pt}
  \renewcommand{\arraystretch}{1.16}
  \begin{tabular}{@{}p{0.13\linewidth}p{0.81\linewidth}@{}}
    \toprule
    \textbf{Oracle} & \textbf{Index and search parameters} \\
    \midrule
    Flat
      & Exact FAISS \texttt{IndexFlatIP} scan; there are no approximation
        parameters. \\
    IVF
      & FAISS \texttt{IndexIVFFlat} with
        $\texttt{nlist}=1{,}024$, $\texttt{nprobe}=96$, and candidate
        replenishment enabled. \\
    HNSW
      & Xbox-transformed FAISS HNSW with $M=48$,
        $\texttt{efConstruction}=500$, and $\texttt{efSearch}=176$. \\
    NSG
      & Xbox-transformed FAISS NSG with graph degree 32,
        construction pool $\texttt{GK}=128$, and search width
        $\texttt{search\_L}=256$.  We use the NN-Descent builder. \\
    MLANN
      & Random-forest index with 24 trees, depth 12, feature density $1/8$,
        minimum vote count 1, and voting threshold $10^{-4.5}$. For training, we sample
32,768 training queries.\\
    PiPNN
      & Xbox-transformed graph with degree 48, $\alpha=1.13$, fanout 30,
        cluster size 1,024, and search beam 256. \\
    \bottomrule
  \end{tabular}
  \caption{Oracle configurations for the 128K SoftmaxLift comparison.}
  \label{tab:oracle-parameters}
\end{table}

\subsection{Large Context GPU Attention Implementation}
At near-million-token context lengths, the complete document KV cache exceeds
  the available GPU memory. To construct an exact GPU reference for our 1M-context
  experiments, we modify the inference implementation of
  Llama-3.1-Nemotron-8B-UltraLong-1M-Instruct to use hybrid CPU--GPU KV-cache
  placement. The cache is stored in BF16 and sharded by transformer layer. As many
  complete layer caches as safely fit in GPU memory are kept resident, while the
  remaining layers are stored in pinned CPU memory and asynchronously prefetched
  into two reusable GPU staging buffers. The transfer for layer $\ell+1$ is
  overlapped with the exact attention computation for layer $\ell$, while the
  question and generated-token KV states remain GPU-resident.

  This implementation changes only the placement of the KV cache; the attention
  computation itself remains exact. It reduces peak GPU memory use and avoids
  host-to-device transfers for resident layers, enabling an exact GPU baseline at
  near-million-token context lengths.

%% file: main.bib
@inproceedings{vaswani2017attention,
  author    = {Vaswani, Ashish and Shazeer, Noam and Parmar, Niki and
               Uszkoreit, Jakob and Jones, Llion and Gomez, Aidan N. and
               Kaiser, {\L}ukasz and Polosukhin, Illia},
  title     = {Attention Is All You Need},
  booktitle = {Advances in Neural Information Processing Systems},
  volume    = {30},
  pages     = {5998--6008},
  year      = {2017},
  url       = {https://proceedings.neurips.cc/paper/2017/hash/3f5ee243547dee91fbd053c1c4a845aa-Abstract.html}
}

@inproceedings{kitaev2020reformer,
  author    = {Kitaev, Nikita and Kaiser, {\L}ukasz and Levskaya, Anselm},
  title     = {Reformer: The Efficient Transformer},
  booktitle = {International Conference on Learning Representations},
  year      = {2020},
  url       = {https://openreview.net/forum?id=rkgNKkHtvB}
}

@inproceedings{zaheer2020bigbird,
  author    = {Zaheer, Manzil and Guruganesh, Guru and Dubey, Avinava and
               Ainslie, Joshua and Alberti, Chris and Onta{\~n}{\'o}n, Santiago and
               Pham, Philip and Ravula, Anirudh and Wang, Qifan and Yang, Li and
               Ahmed, Amr},
  title     = {Big Bird: Transformers for Longer Sequences},
  booktitle = {Advances in Neural Information Processing Systems},
  volume    = {33},
  pages     = {17283--17297},
  year      = {2020},
  url       = {https://proceedings.neurips.cc/paper/2020/hash/c8512d142a2d849725f31a9a7a361ab9-Abstract.html}
}

@inproceedings{shrivastava2014alsh,
  author    = {Shrivastava, Anshumali and Li, Ping},
  title     = {Asymmetric {LSH} ({ALSH}) for Sublinear Time Maximum Inner Product
               Search ({MIPS})},
  booktitle = {Advances in Neural Information Processing Systems},
  volume    = {27},
  pages     = {2321--2329},
  year      = {2014},
  url       = {https://proceedings.neurips.cc/paper/2014/hash/c98e7c4b8f20d384e3ad857d0ee226cc-Abstract.html}
}

@inproceedings{neyshabur2015symmetric,
  author    = {Neyshabur, Behnam and Srebro, Nathan},
  title     = {On Symmetric and Asymmetric {LSH}s for Inner Product Search},
  booktitle = {Proceedings of the 32nd International Conference on Machine Learning},
  series    = {Proceedings of Machine Learning Research},
  volume    = {37},
  pages     = {1926--1934},
  publisher = {PMLR},
  year      = {2015},
  url       = {https://proceedings.mlr.press/v37/neyshabur15.html}
}

@inproceedings{retrievalattention,
  author    = {Liu, Di and Chen, Meng and Lu, Baotong and Jiang, Huiqiang and
               Han, Zhenhua and Zhang, Qianxi and Chen, Qi and Zhang, Chengruidong and
               Ding, Bailu and Zhang, Kai and Chen, Chen and Yang, Fan and
               Yang, Yuqing and Qiu, Lili},
  title     = {{RetrievalAttention}: Accelerating Long-Context {LLM} Inference via
               Vector Retrieval},
  booktitle = {Advances in Neural Information Processing Systems},
  volume    = {38},
  year      = {2025},
  doi       = {10.52202/085713-1816},
  url       = {https://proceedings.neurips.cc/paper_files/paper/2025/hash/4e36d4049fb0fea195a8267c8dcd0824-Abstract-Conference.html}
}

@inproceedings{squeezedattention,
  author    = {Hooper, Coleman Richard Charles and Kim, Sehoon and
               Mohammadzadeh, Hiva and Maheswaran, Monishwaran and
               Zhao, Sebastian and Paik, June and Mahoney, Michael W. and
               Keutzer, Kurt and Gholami, Amir},
  title     = {Squeezed Attention: Accelerating Long Context Length {LLM} Inference},
  booktitle = {Proceedings of the 63rd Annual Meeting of the Association for
               Computational Linguistics (Volume 1: Long Papers)},
  pages     = {32631--32652},
  publisher = {Association for Computational Linguistics},
  year      = {2025},
  doi       = {10.18653/v1/2025.acl-long.1568},
  url       = {https://aclanthology.org/2025.acl-long.1568/}
}

@inproceedings{hashattention,
  author    = {Desai, Aditya and Yang, Shuo and Cuadron, Alejandro and
               Zaharia, Matei and Gonzalez, Joseph E. and Stoica, Ion},
  title     = {{HashAttention}: Semantic Sparsity for Faster Inference},
  booktitle = {Proceedings of the 42nd International Conference on Machine Learning},
  series    = {Proceedings of Machine Learning Research},
  volume    = {267},
  pages     = {13402--13418},
  publisher = {PMLR},
  year      = {2025},
  url       = {https://proceedings.mlr.press/v267/desai25a.html}
}

@misc{saap,
  author        = {Mazar{\'e}, Pierre-Emmanuel and Szilvasy, Gergely and
                   Lomeli, Maria and Massa, Francisco and Murray, Naila and
                   J{\'e}gou, Herv{\'e} and Douze, Matthijs},
  title         = {Inference-Time Sparse Attention with Asymmetric Indexing},
  year          = {2025},
  eprint        = {2502.08246},
  archiveprefix = {arXiv},
  primaryclass  = {cs.CL},
  url           = {https://arxiv.org/abs/2502.08246}
}

@inproceedings{vattention,
  author    = {Desai, Aditya and Agrawal, Kumar Krishna and Yang, Shuo and
               Cuadron Lafuente, Alejandro and Schroeder, Luis Gaspar and
               Zaharia, Matei and Gonzalez, Joseph E. and Stoica, Ion},
  title     = {{vAttention}: Verified Sparse Attention via Sampling},
  booktitle = {The Fourteenth International Conference on Learning Representations},
  year      = {2026},
  url       = {https://openreview.net/forum?id=zzTDulLys0}
}

@inproceedings{magicpig,
  author    = {Chen, Zhuoming and Sadhukhan, Ranajoy and Ye, Zihao and Zhou, Yang and
               Zhang, Jianyu and Nolte, Niklas and Tian, Yuandong and Douze, Matthijs and
               Bottou, L{\'e}on and Jia, Zhihao and Chen, Beidi},
  title     = {{MagicPIG}: {LSH} Sampling for Efficient {LLM} Generation},
  booktitle = {The Thirteenth International Conference on Learning Representations},
  year      = {2025},
  url       = {https://openreview.net/forum?id=ALzTQUgW8a}
}

@article{aumuller2020annbenchmarks,
  author  = {Aum{\"u}ller, Martin and Bernhardsson, Erik and Faithfull, Alexander},
  title   = {{ANN}-Benchmarks: A Benchmarking Tool for Approximate Nearest Neighbor
             Algorithms},
  journal = {Information Systems},
  volume  = {87},
  pages   = {101374},
  year    = {2020},
  doi     = {10.1016/j.is.2019.02.006},
  url     = {https://doi.org/10.1016/j.is.2019.02.006}
}

@article{jaasaari2026vibe,
  author  = {J{\"a}{\"a}saari, Elias and Hyv{\"o}nen, Ville and Ceccarello, Matteo and
             Roos, Teemu and Aum{\"u}ller, Martin},
  title   = {{VIBE}: Vector Index Benchmark for Embeddings},
  journal = {Journal of Data-centric Machine Learning Research},
  year    = {2026},
  url     = {https://openreview.net/forum?id=6Sx6ra1QLz}
}

@article{duffield2007priority,
  author  = {Nick Duffield and Carsten Lund and Mikkel Thorup},
  title   = {Priority Sampling for Estimation of Arbitrary Subset Sums},
  journal = {Journal of the ACM},
  volume  = {54},
  number  = {6},
  year    = {2007},
  doi     = {10.1145/1314690.1314696}
}

@inproceedings{daliri2024simple,
  author    = {Daliri, Majid and Freire, Juliana and Musco, Christopher and
               Santos, A{\'e}cio and Zhang, Haoxiang},
  title     = {Simple Analysis of Priority Sampling},
  booktitle = {2024 Symposium on Simplicity in Algorithms ({SOSA})},
  pages     = {224--229},
  publisher = {Society for Industrial and Applied Mathematics},
  year      = {2024},
  doi       = {10.1137/1.9781611977936.21},
  url       = {https://doi.org/10.1137/1.9781611977936.21}
}

@book{tsybakov2009,
  author    = {Alexandre B. Tsybakov},
  title     = {Introduction to Nonparametric Estimation},
  series    = {Springer Series in Statistics},
  publisher = {Springer},
  address   = {New York},
  year      = {2009},
  doi       = {10.1007/b13794},
  isbn      = {978-0-387-79051-0}
}

@article{child2019sparsetransformer,
  title  = {Generating Long Sequences with Sparse Transformers},
  author = {Child, Rewon and Gray, Scott and Radford, Alec and Sutskever, Ilya},
  journal = {arXiv preprint arXiv:1904.10509},
  year   = {2019}
}

@article{beltagy2020longformer,
  title  = {Longformer: The Long-Document Transformer},
  author = {Beltagy, Iz and Peters, Matthew E. and Cohan, Arman},
  journal = {arXiv preprint arXiv:2004.05150},
  year   = {2020}
}

@inproceedings{xiao2024streamingllm,
  title     = {Efficient Streaming Language Models with Attention Sinks},
  author    = {Xiao, Guangxuan and Tian, Yuandong and Chen, Beidi and Han, Song and Lewis, Mike},
  booktitle = {International Conference on Learning Representations (ICLR)},
  year      = {2024}
}

@inproceedings{zhang2023h2o,
  title     = {{H2O}: Heavy-Hitter Oracle for Efficient Generative Inference of Large Language Models},
  author    = {Zhang, Zhenyu and Sheng, Ying and Zhou, Tianyi and Chen, Tianlong and Zheng, Lianmin
             and Cai, Ruisi and Song, Zhao and Tian, Yuandong and R{\'e}, Christopher
             and Barrett, Clark and Wang, Zhangyang and Chen, Beidi},
  booktitle = {Advances in Neural Information Processing Systems (NeurIPS)},
  year      = {2023}
}

@inproceedings{liu2023scissorhands,
  title     = {Scissorhands: Exploiting the Persistence of Importance Hypothesis for {LLM} {KV} Cache
               Compression at Test Time},
  author    = {Liu, Zichang and Desai, Aditya and Liao, Fangshuo and Wang, Weitao and Xie, Victor
             and Xu, Zhaozhuo and Kyrillidis, Anastasios and Shrivastava, Anshumali},
  booktitle = {Advances in Neural Information Processing Systems (NeurIPS)},
  year      = {2023}
}

@inproceedings{ge2024fastgen,
  title     = {Model Tells You What to Discard: Adaptive {KV} Cache Compression for {LLM}s},
  author    = {Ge, Suyu and Zhang, Yunan and Liu, Liyuan and Zhang, Minjia and Han, Jiawei and Gao, Jianfeng},
  booktitle = {International Conference on Learning Representations (ICLR)},
  year      = {2024}
}

@inproceedings{li2024snapkv,
  title     = {{SnapKV}: {LLM} Knows What You Are Looking For Before Generation},
  author    = {Li, Yuhong and Huang, Yingbing and Yang, Bowen and Venkitesh, Bharat and Locatelli, Acyr
             and Ye, Hanchen and Cai, Tianle and Lewis, Patrick and Chen, Deming},
  booktitle = {Advances in Neural Information Processing Systems (NeurIPS)},
  year      = {2024}
}

@inproceedings{tang2024quest,
  title     = {{Quest}: Query-Aware Sparsity for Efficient Long-Context {LLM} Inference},
  author    = {Tang, Jiaming and Zhao, Yilong and Zhu, Kan and Xiao, Guangxuan and Kasikci, Baris and Han, Song},
  booktitle = {International Conference on Machine Learning (ICML)},
  year      = {2024}
}

@inproceedings{xiao2024infllm,
  title     = {{InfLLM}: Unveiling the Intrinsic Capacity of {LLM}s for Understanding Extremely Long
               Sequences with Training-Free Memory},
  author    = {Xiao, Chaojun and Zhang, Pengle and Han, Xu and Xiao, Guangxuan and Lin, Yankai
             and Zhang, Zhengyan and Liu, Zhiyuan and Sun, Maosong},
  booktitle = {Advances in Neural Information Processing Systems (NeurIPS)},
  year      = {2024}
}

@article{yang2024doublesparsity,
  title   = {Post-Training Sparse Attention with Double Sparsity},
  author  = {Yang, Shuo and Sheng, Ying and Gonzalez, Joseph E. and Stoica, Ion and Zheng, Lianmin},
  journal = {arXiv preprint arXiv:2408.07092},
  year    = {2024}
}

@inproceedings{singhania2024loki,
  title     = {{Loki}: Low-Rank Keys for Efficient Sparse Attention},
  author    = {Singhania, Prajwal and Singh, Siddharth and He, Shwai and Feizi, Soheil and Bhatele, Abhinav},
  booktitle = {Advances in Neural Information Processing Systems (NeurIPS)},
  year      = {2024}
}

@inproceedings{yuan2025nsa,
  title     = {Native Sparse Attention: Hardware-Aligned and Natively Trainable Sparse Attention},
  author    = {Yuan, Jingyang and Gao, Huazuo and Dai, Damai and Luo, Junyu and Zhao, Liang
             and Zhang, Zhengyan and Xie, Zhenda and Wei, Y. X. and Wang, Lean and Xiao, Zhiping
             and Wang, Yuqing and Ruan, Chong and Zhang, Ming and Liang, Wenfeng and Zeng, Wangding},
  booktitle = {Annual Meeting of the Association for Computational Linguistics (ACL)},
  year      = {2025}
}

@article{zhang2024pqcache,
  title   = {{PQCache}: Product Quantization-based {KVCache} for Long Context {LLM} Inference},
  author  = {Zhang, Hailin and Ji, Xiaodong and Chen, Yilin and Fu, Fangcheng and Miao, Xupeng
           and Nie, Xiaonan and Chen, Weipeng and Cui, Bin},
  journal = {Proceedings of the ACM on Management of Data},
  volume  = {3},
  number  = {3},
  year    = {2025}
}

@article{zhu2025tactic,
  title   = {{Tactic}: Adaptive Sparse Attention with Clustering and Distribution Fitting for
             Long-Context {LLM}s},
  author  = {Zhu, Kan and Tang, Tian and Xu, Qinyu and Gu, Yile and Zeng, Zhichen and Kadekodi, Rohan
           and Zhao, Liangyu and Li, Ang and Krishnamurthy, Arvind and Kasikci, Baris},
  journal = {arXiv preprint arXiv:2502.12216},
  year    = {2025}
}

@inproceedings{bachrach2014xbox,
  title     = {Speeding Up the {Xbox} Recommender System Using a {E}uclidean Transformation for
               Inner-Product Spaces},
  author    = {Bachrach, Yoram and Finkelstein, Yehuda and Gilad-Bachrach, Ran and Katzir, Liran
             and Koenigstein, Noam and Nice, Nir and Paquet, Ulrich},
  booktitle = {ACM Conference on Recommender Systems (RecSys)},
  pages     = {257--264},
  year      = {2014}
}

@article{jaiswal2022ooddiskann,
  title   = {{OOD-DiskANN}: Efficient and Scalable Graph {ANNS} for Out-of-Distribution Queries},
  author  = {Jaiswal, Shikhar and Krishnaswamy, Ravishankar and Garg, Ankit
           and Simhadri, Harsha Vardhan and Agrawal, Sheshansh},
  journal = {arXiv preprint arXiv:2211.12850},
  year    = {2022}
}

@article{chen2024roargraph,
  title   = {{RoarGraph}: A Projected Bipartite Graph for Efficient Cross-Modal Approximate Nearest
             Neighbor Search},
  author  = {Chen, Meng and Zhang, Kai and He, Zhenying and Jing, Yinan and Wang, X. Sean},
  journal = {Proceedings of the VLDB Endowment},
  volume  = {17},
  number  = {11},
  year    = {2024}
}

@inproceedings{karppa2022deann,
  title     = {{DEANN}: Speeding Up Kernel-Density Estimation Using Approximate Nearest Neighbor Search},
  author    = {Karppa, Matti and Aum{\"u}ller, Martin and Pagh, Rasmus},
  booktitle = {International Conference on Artificial Intelligence and Statistics (AISTATS)},
  pages     = {3108--3137},
  year      = {2022}
}

@inproceedings{han2024hyperattention,
  title     = {{HyperAttention}: Long-Context Attention in Near-Linear Time},
  author    = {Han, Insu and Jayaram, Rajesh and Karbasi, Amin and Mirrokni, Vahab
             and Woodruff, David P. and Zandieh, Amir},
  booktitle = {International Conference on Learning Representations (ICLR)},
  year      = {2024}
}

@article{schroder2026wildcat,
  title={Wildcat: Near-linear attention in theory and practice},
  author={Schr{\"o}der, Tobias and Mackey, Lester},
  journal={arXiv preprint arXiv:2602.10056},
  year={2026}
}

@article{kochetkova2026streaming,
  title={Streaming attention approximation via discrepancy theory},
  author={Kochetkova, Ekaterina and Sheth, Kshiteej Jitesh and Han, Insu and Zandieh, Amir and Kapralov, Michael},
  journal={Advances in Neural Information Processing Systems},
  volume={38},
  pages={20935--20972},
  year={2026}
}

@misc{liberty2026coresets,
      title={Nearly Optimal Attention Coresets}, 
      author={Edo Liberty and Alexandr Andoni and Eldar Kleiner},
      year={2026},
      eprint={2605.05602},
      archivePrefix={arXiv},
      primaryClass={cs.DS},
      url={https://arxiv.org/abs/2605.05602}, 
}

@article{ohlsson1998sequential,
  title   = {Sequential {P}oisson Sampling},
  author  = {Ohlsson, Esbj{\"o}rn},
  journal = {Journal of Official Statistics},
  volume  = {14},
  number  = {2},
  pages   = {149--162},
  year    = {1998}
}

@article{horvitz1952generalization,
  title   = {A Generalization of Sampling Without Replacement from a Finite Universe},
  author  = {Horvitz, Daniel G. and Thompson, Donovan J.},
  journal = {Journal of the American Statistical Association},
  volume  = {47},
  number  = {260},
  pages   = {663--685},
  year    = {1952}
}

@inproceedings{alon2005subsetsums,
  title     = {Estimating Arbitrary Subset Sums with Few Probes},
  author    = {Alon, Noga and Duffield, Nick and Lund, Carsten and Thorup, Mikkel},
  booktitle = {ACM Symposium on Principles of Database Systems (PODS)},
  pages     = {317--325},
  year      = {2005}
}

@inproceedings{szegedy2006priority,
  title     = {The {DLT} Priority Sampling Is Essentially Optimal},
  author    = {Szegedy, Mario},
  booktitle = {ACM Symposium on Theory of Computing (STOC)},
  pages     = {150--158},
  year      = {2006}
}

@article{daliri2024inner,
  title   = {Sampling Methods for Inner Product Sketching},
  author  = {Daliri, Majid and Freire, Juliana and Musco, Christopher and Santos, A{\'e}cio
           and Zhang, Haoxiang},
  journal = {Proceedings of the VLDB Endowment},
  volume  = {17},
  number  = {9},
  pages   = {2185--2197},
  year    = {2024}
}

@inproceedings{mussmann2016mips,
  title     = {Learning and Inference via Maximum Inner Product Search},
  author    = {Mussmann, Stephen and Ermon, Stefano},
  booktitle = {International Conference on Machine Learning (ICML)},
  pages     = {2587--2596},
  year      = {2016}
}

@inproceedings{mussmann2017fast,
  title     = {Fast Amortized Inference and Learning in Log-Linear Models with Randomly Perturbed
               Nearest Neighbor Search},
  author    = {Mussmann, Stephen and Levy, Daniel and Ermon, Stefano},
  booktitle = {Conference on Uncertainty in Artificial Intelligence (UAI)},
  year      = {2017}
}

@inproceedings{maddison2014astar,
  title     = {A* Sampling},
  author    = {Maddison, Chris J. and Tarlow, Daniel and Minka, Tom},
  booktitle = {Advances in Neural Information Processing Systems (NeurIPS)},
  year      = {2014}
}

@inproceedings{kool2019gumbel,
  title     = {Stochastic Beams and Where to Find Them: The {G}umbel-Top-$k$ Trick for Sampling
               Sequences Without Replacement},
  author    = {Kool, Wouter and van Hoof, Herke and Welling, Max},
  booktitle = {International Conference on Machine Learning (ICML)},
  year      = {2019}
}

@article{spring2017lsh,
  title   = {A New Unbiased and Efficient Class of {LSH}-Based Samplers and Estimators for Partition
             Function Computation in Log-Linear Models},
  author  = {Spring, Ryan and Shrivastava, Anshumali},
  journal = {arXiv preprint arXiv:1703.05160},
  year    = {2017}
}

@inproceedings{luo2018ace,
  title     = {Arrays of (Locality-Sensitive) Count Estimators ({ACE}): Anomaly Detection on the Edge},
  author    = {Luo, Chen and Shrivastava, Anshumali},
  booktitle = {The Web Conference (WWW)},
  pages     = {1439--1448},
  year      = {2018}
}

@inproceedings{gupta2021topk,
  title     = {Memory-Efficient Transformers via Top-$k$ Attention},
  author    = {Gupta, Ankit and Dar, Guy and Goodman, Shaya and Ciprut, David and Berant, Jonathan},
  booktitle = {Workshop on Simple and Efficient Natural Language Processing (SustaiNLP)},
  year      = {2021}
}

@article{mussmann2026sum,
  title   = {Sum Estimation via Vector Similarity Search},
  author  = {Mussmann, Stephen and Raje, Mehul Smriti and Tumkur, Kavya and Messoussi, Oumayma and Hachem, Cyprien and Jacob, Seby},
  journal = {arXiv preprint arXiv:2601.11765},
  year    = {2026},
  doi     = {10.48550/arXiv.2601.11765},
  url     = {https://arxiv.org/abs/2601.11765}
}

@inproceedings{hsieh2024ruler,
  title     = {{RULER}: What's the Real Context Size of Your Long-Context Language Models?},
  author    = {Hsieh, Cheng-Ping and Sun, Simeng and Kriman, Samuel and Acharya, Shantanu and Rekesh, Dima and Jia, Fei and Zhang, Yang and Ginsburg, Boris},
  booktitle = {First Conference on Language Modeling (COLM)},
  year      = {2024}
}
